\documentclass[11pt,letterpaper]{article}
\usepackage[utf8]{inputenc}
\usepackage{docmute}
\usepackage{appendix}
\usepackage[english]{babel}
\usepackage{amsmath}
\usepackage{amsthm}
\usepackage{amsfonts}
\usepackage{amssymb}
\usepackage{graphicx}
\usepackage{afterpage}
\usepackage{setspace}

\usepackage[top=1.25in, bottom=1.25in, left=1.25in, right=1.25in]{geometry}

\newcommand{\rank}{\mathrm{rank}}
\newcommand{\diag}{\mathrm{diag}}

\numberwithin{equation}{section}

\newtheorem{theorem}{Theorem}
\newtheorem{lemma}{Lemma}

\newtheorem{assumption}{Assumption}

\newtheorem{remark}{Remark}
\usepackage{bm}

\usepackage{hyperref}
\PassOptionsToPackage{linktocpage}{hyperref}

\usepackage{breakcites}
\usepackage[authoryear, round]{natbib}
\title{Identification and estimation of treatment effects in a linear factor model with fixed number of time periods\thanks{The study is supported by JSPS KAKENHI Grant Numbers JP23K18799 (Fusejima); and JP22K13373 (Ishihara).
We thank Yukitoshi Matsushita, Katsumi Shimotsu, Kohei Yata, and seminar participants at Hitotsubashi University, The University of Tokyo, Tohoku University, and Kansai Keiryo Keizaigaku Kenkyukai.}}

\author{Koki Fusejima\thanks{Hitotsubashi Institute for Advanced Study, Hitotsubashi University; 2-1 Naka, Kunitachi, Tokyo 186-8601, Japan; k.fusejima@r.hit-u.ac.jp} and
Takuya Ishihara\thanks{Graduate School of Economics and Management, Tohoku University; Email: takuya.ishihara.b7@tohoku.ac.jp}}

\date{This version: \today}

\begin{document}

\maketitle

\begin{abstract}
This paper provides a new approach for identifying and estimating the Average Treatment Effect on the Treated under a linear factor model that allows for multiple time-varying unobservables.
Unlike the majority of the literature on treatment effects in linear factor models, our approach does not require the number of pre-treatment periods to go to infinity to obtain a valid estimator.
Our identification approach employs a certain nonlinear transformations of the time invariant observed covariates that are sufficiently correlated with the unobserved variables.
This relevance condition can be checked with the available data on pre-treatment periods by validating the correlation of the transformed covariates and the pre-treatment outcomes.
Based on our identification approach, we provide an asymptotically unbiased estimator of the effect of participating in the treatment when there is only one treated unit and the number of control units is large.
\end{abstract}

\newpage

\section{Introduction}\label{3-sec:1}
The use of panel data has been increasingly popular in empirical socio-economic studies. 
An essential advantage of panel data is that researchers can obtain consistent estimates of important parameters while controlling for unobserved heterogeneity.
The most common version of this approach for identifying the causal effect of a binary treatment on an outcome of interest is the difference in differences (DID) method.
The main motivating model for the DID approach is one in which untreated ‘‘potential’’ outcomes are generated by a two-way fixed effects model, a model that allows for unobserved time invariant individual and time fixed effects. 
This model is consistent with the so-called ‘‘parallel trends’’ assumption underlying the DID approach (see, for example, \cite{callaway2022difference} for a recent review on the DID approach).
This assumption requires that the treated group's counterfactual path of untreated potential outcomes is the same as the path of actual outcomes for the untreated group.
However, the parallel trends assumption may only hold in some applications. The main concern with this assumption is that the effect of some time invariant unobserved variable may change over time and therefore cause untreated potential outcomes to follow different paths for the treated group relative to the untreated group.

In this paper, we develop a new approach for identifying and estimating the Average Treatment Effect on the Treated (ATT) when untreated potential outcomes are generated by a linear factor model (also known as an interactive fixed effects model).
This model allows for unobserved multiple time-varying individual effects.
We establish the identification of the ATT when a researcher observes a vector of time invariant covariates.
Based on our identification result, we propose an estimation method that is asymptotically valid when the number of control units goes to infinity.
There is a growing literature on treatment effects in linear factor models, and one of the most popular methods is the Synthetic Control (SC) method, proposed in a series of influential papers by \cite{abadie2003economic}, \cite{abadie2010synthetic}, and \cite{abadie2015comparative}.
Unlike the majority of the literature, such as synthetic controls and more recent methods based on the original SC estimator (see, for example, \cite{abadie2021using} for a recent review on the SC method), our approach does not require the number of pre-treatment periods to go to infinity to obtain an asymptotically unbiased estimator of the effect of participating in the treatment.
In many applications of SC and related methods to comparative case studies and micro-data, the number of pre-treatment periods may not be large enough for asymptotic justification, and the number of control units may be as large or larger than that of pre-treatment periods (e.g., \cite{doudchenko2016balancing}).

The key assumption of our identification approach is that certain nonlinear transformations of the observed covariates are sufficiently correlated with the unobserved variables.
Although the effects of the covariates and the unobserved variables are not identified, we show that the covariates help identify the effect of participating in the treatment.
We can check the relevance condition with the available data by validating the correlation of the transformed covariates and the pre-treatment outcomes instead of the unobserved variables.

The basic idea of our identification approach is to use the transformed covariates as instruments for the unobserved variables.
The corresponding instrumental variable estimator is infeasible due to the unobserved variables.
However, we need not recover the effects of these variables to identify the effect of participating in the treatment.
We show that, when we have access to pre-treatment periods, it suffices to recover the effects of the pre-treatment outcomes in a regression where the individual fixed effects are eliminated.
In our approach, we employ an alternative estimator that replaces the unobserved variables with the pre-treatment outcomes for recovering the effect of participating in the treatment.

We develop estimators of the ATT using our identification approach. 
We propose a two-step estimator where, in the first step, we estimate the effects of the pre-treatment outcomes in a differenced out model.
In the second step, we plug these in to obtain estimates of the ATT.
For estimation, we focus on the case when only one treated unit exists. This case is the main application of SC methods to comparative case studies.
The SC method and our approach are based on opposite ideas.
SC method is based on the idea that a weighted average of the outcomes of the control units reconstructs the pre-treatment outcomes of the treated unit and use the pre-treatment periods to estimate these weights.
Our approach is based on the idea that a linear combination of the pre-treatment outcomes recovers the outcomes after the treatment period and uses the control units to estimate the corresponding parameters.

This paper is also related to works on linear factor models with a small number of periods.
The closely related studies are \cite{ahn2013panel}, \cite{imbens2021controlling}, \cite{callaway2022treatment}, and \cite{brown2022unified}; similar to the latter three studies, we generalize the DID-type approach and obtain identification by providing some additional conditions.
\cite{imbens2021controlling} employ post-treatment observations under an additional serial independence assumption, and
\cite{callaway2022treatment} employ some covariates that have time invariant effects on untreated potential outcomes to recover average causal effects.
Our approach may have broader applicability than these methods.
Because our approach does not require knowledge of post-treatment observations, our approach can be applied just after the timing of treatment.
Although our estimator is constructed to have good properties for short periods, our estimator can also be applied when the number of periods is large enough for all the covariates to have time varying effects on untreated potential outcomes.
Recently, \cite{brown2022unified} propose a general framework for incorporating factor-model estimators for treatment effect estimation with a small number of periods.
Their framework is more general and similar to the identification approach of \cite{callaway2022treatment} and \cite{imbens2021controlling}.
Our identification approach is also similar to that of \cite{brown2022unified}, but differs in how we use covariates.
For factor-model estimation, they mainly focus on the estimation approach of \cite{ahn2013panel}, who propose an estimator of linear factor models with a small number of periods.
The estimation approach of \cite{ahn2013panel} allows general instruments as well as time varying covariates and is more general than our approach.
Our approach focus on the case where we have time invariant covariates that is included in the model and we employ the linear structure for identifying the treatment effect.

The outline of the paper is as follows.
Section \ref{3-sec:2} introduces the linear factor model and our basic assumptions.
Section \ref{3-sec:3} provides our main identification arguments.  
Section \ref{3-sec:4} proposes an estimation method based on the identification results.
Section \ref{3-sec:5} provides Monte Carlo simulations that compare our approach to DID and SC methods.
Section \ref{3-sec:6} concludes. 
We provide the the proofs of the results in the Appendix.

\section{Settings}\label{3-sec:2}
In this section, we introduce the linear factor model and our basic assumptions.
Suppose we have a balanced panel of $N$ units observed on a total of $T$ periods.
Let $Y_{it}\in \mathbb{R}$ denote an observed outcome in a particular period $t$.
Individuals either belong to a treated group or an untreated group.
We set $D_i\in\{0,1\}$ as a treatment indicator so that $D_i = 1$ for individuals in the treated group and $D_i = 0$ for the untreated group.
We assume that treatment occurs in period 0, which is the same across all individuals, and we have access to $T_0$ and $T_1+1$ periods of data indexed by $t=-T_0,\ldots,-1$ and $t=0,\ldots,T_1$ before and after the policy change, respectively.
Individuals have treated potential outcomes and untreated potential outcomes at each period. 
We denote these variables by $Y_{it}(1)$ and $Y_{it}(1)$, respectively, for $t = -T_0,\ldots,T_1$.
In pre-treatment periods, we have $Y_{it} = Y_{it} (0)$
for all $i$ and $t < 0$; that is, we observe untreated potential outcomes for all individuals in these periods. 
When $t \geq 0$, we have $Y_{it} = D_i Y_{it}(1) +(1-D_i) Y_{it}(0)$; that is, in period 0 and subsequent periods, we observe treated potential outcomes for individuals in the treated group and untreated potential outcomes for individuals in the untreated group.

We also suppose that we observe a vector of time-invariant covariates $Z_i\in\mathbb{R}^d$.
A key difference between our approach and standard panel data models is that we treat the treatment variable asymmetrically from the other covariates.
Unlike the traditional estimation approach, the objective in the treatment effects literature is not to estimate the effect of the covariates but to control for them. 

The parameter of interest is the Average Treatment Effect on the Treated (ATT); the average treatment effect on the treated conditional on $Z=z$ in period $t$ is
\begin{equation*}
E[Y_{it}(1) - Y_{it}(0)|D_i=1, Z_i=z]
= E[Y_{it}|D_i=1, Z_i=z] - E[Y_{it}(0)|D_i=1, Z_i=z].
\end{equation*}
Because $E[Y_{it}|D_i=1, Z_i=z]$ is immediately identified from the sampling process, we consider the identification and estimation of $E[Y_{it}(0)|D_i=1, Z_i=z]$ for $t = 0, \ldots, T_1$ when $(T_0,T_1)$ is fixed.

We first impose the following structure on the data generating process in each period.
\begin{assumption}\label{asp:1}
The data $(\{Y_{it}(0), Y_{it}(1)\}_{t=-T_0}^{T_1},D_i,Z_i)_{i=0}^{N-1}$ is independent over $i$ and satisfies the following model:
For $i=0,1,\ldots,N-1$,
\begin{eqnarray}
Y_{it} &=& 
\begin{cases}
D_i Y_{it}(1) +(1-D_i) Y_{it}(0), & t=0,\ldots,T_1, \\
Y_{it}(0), & t=-T_0,\ldots,-1,
\end{cases}, \label{model-1} \\
Y_{it}(0) &=& b_t' Z_i + F_t' \lambda_i + \epsilon_{it}, \ \ \ t = -T_0. \ldots, T_1, \label{model-2},
\end{eqnarray}
where $b_t\in\mathbb{R}^d$ is a vector of structural parameters, $\lambda_i\in\mathbb{R}^r$ is a vector of factor loadings, $F_t\in\mathbb{R}^r$ is a vector of common factors, and $\epsilon_{it}\in\mathbb{R}$ is an idiosyncratic error. 
We treat $F_t$ as a non-random vector.
\end{assumption}
\noindent We assume that in all periods, untreated potential outcomes are generated by a linear factor model with multiple time invariant unobservables whose ‘‘effects’’ can change over time.
\footnote{We can include an individual fixed effect into the term $F_t' \lambda_i$ by having $\lambda_i$ include an element where the corresponding element of $F_t$ is equal to one.
Similarly, we can include a time fixed effect into the term $b_t' Z_i$ by having $Z_i$ include an intercept where the corresponding element of $b_t$ is equal to the time fixed effect.}
The SC methods also employ linear factor models to establish their validity.

We highlight some important comments on the model in Assumption \ref{asp:1}.
First, Assumption \ref{asp:1} only imposes structure on how untreated potential outcomes are generated and allows for essentially unrestricted treatment effect heterogeneity.
This structure is standard in the literature on treatment effects with panel data.
Second, Assumption \ref{asp:1} allows the distributions of $\lambda_i$ and $Z_i$ to be different for individuals in the treated and untreated groups.
As in fixed effects models, Assumption \ref{asp:1} also allows for $\lambda_i$ and $Z_i$ to be arbitrarily correlated.

Contrary to the DID literature, the linear factor structure in \eqref{model-2} allows the effect of the unobserved $\lambda_i$ to vary over time.
Suppose one removes the linear factor structure from the model in Assumption \ref{asp:1}. 
In that case, this model is a two-way fixed effects model consistent with the conditional parallel trends assumptions.
Therefore, if we observe both $\lambda_i$ and $Z_i$, we could use the conditional DID approach.
Because we do not observe $\lambda_i$, we need to make some modifications.

We also make the following assumptions throughout the paper:
\begin{assumption}\label{asp:2}
We have $E[\epsilon_{it}|D_i,Z_i]=0$, $E[Z_i Z_i'|D_i=0]$ is full rank, and the support of $(D_i,Z_i)$ is $\{0,1\} \times \mathcal{Z}$, where $\mathcal{Z}$ is the support of $Z_i$.
\end{assumption}
\noindent First, Assumption \ref{asp:2} assumes that the idiosyncratic error $\epsilon_{it}$ is mean independent of $D_i$ and $Z_i$. 
In the literature on linear factor models for causal effects, many papers assume $E[\epsilon_{it}|\lambda_i, D_i, Z_i]=0$; i.e., $\epsilon_{it}$ is mean independent of $\lambda_i$ as well as $D_i$ and $Z_i$ (e.g., \cite{gobillon2016regional}, \cite{freyberger2018non}, and \cite{imbens2021controlling}.) 
Our approach allows an arbitrary correlation between $\epsilon_{it}$ and $\lambda_i$.
Second, Assumption \ref{asp:2} assumes that the covariates $Z_i$ have no perfect multicollinearity and that the support of $Z_i$ is sufficiently large for both treated and untreated groups.
These requirements are standard in linear regression models.

\section{Identification}\label{3-sec:3}

\subsection{Intuition of our identification challenge}\label{3-sec:3.1}
In this section, we outline the challenges that we face in identifying ATT when untreated potential outcomes are generated by the interactive fixed effects model in Assumption \ref{asp:1}. 
To explain the intuition of our identification challenge, we consider a simple case with no covariates where $b_t = 0$ for all $t$.
For this case, \cite{brown2022unified} proposed a general identification scheme.
Although the following explanation can be seen as interpreting the framework presented in \cite{brown2022unified}, we discuss this here for completeness.

Let $\bm{Y}_{i, pre} \equiv \left( Y_{i,-1}, \ldots, Y_{i,-T_0} \right)'$, $\bm{\epsilon}_{i,pre} \equiv (\epsilon_{i,-1}, \ldots , \epsilon_{i,-T_0})'$, and $\bm{F} \equiv (F_{-1}, \ldots , F_{-T_0})$.
Assume that we have at least $r$ pre-treatment periods, i.e., $T_0\geq r$, and that $\bm{F}$ is full row rank, which requires there is sufficient variation in $F_t$ in the pre-treatment periods.
Because no one is treated before period 0, we observe untreated potential outcomes in the pre-treatment periods for both treated and untreated individuals, and $\bm{Y}_{i, pre}$ can be written as follows:
\begin{equation}\label{eq:2.1}
\bm{Y}_{i, pre} \ = \ \bm{F}' \lambda_i + \bm{\epsilon}_{i,pre}.
\end{equation}
If $T_0\geq r$ and $\bm{F}$ is full row rank, we can obtain an expression for $\lambda_i$ by re-arranging the terms in \eqref{eq:2.1}:
\begin{equation}\label{eq:2.2}
\lambda_i = (\bm F')^{+}(\bm Y_{i,pre} - \bm \epsilon_{i,pre}),
\end{equation}
where for any matrix $\bm{B}$, $\bm{B}^{+}$ is the Moore-Penrose generalized inverse matrix of $\bm{B}$. 
Plugging in the expression for $\lambda_i$ in \eqref{eq:2.2} to \eqref{model-2} and re-arranging terms give, for all $t=0,\ldots,T_1$,
\begin{equation}\label{eq:2.3}
Y_{it}(0) = F^{*'}_t \bm Y_{i,pre} + \epsilon_{it} - (F^{*}_t)' \bm \epsilon_{i,pre},
\end{equation}
where we define $F^*_t = \bm F^+F_t$.
Note that \eqref{eq:2.3} holds for both the treated and untreated groups.
Consequently, we can eliminate the time invariant unobservables from the expression of $Y_{it}(0)$ if we know the value of $F^*_t$.

Then, for identifying ATT, notice that by Assumptions \ref{asp:1} and \ref{asp:2}, 
\begin{equation}\label{eq:2.4}
E[Y_{it}(0)|D_i=1] \ = \ (F^*_t)' E[\bm{Y}_{i, pre}|D_i=1]
\end{equation}
holds for any time period $t\geq0$ because $F^*_t$ is non-random.
Hence, if $F_{t}^*$ is identified, then $E[Y_{it}(0)|D_i=1]$ is also identified because all the other terms in the right-hand side of \eqref{eq:2.4} are immediately identified from the sampling process.

The above argument suggests that the key identification challenge is to recover the parameter $F^*_t$.
It is also worth pointing out that identifying $F^*_t$ is easier than identifying $F_t$ itself in all the periods.
In particular, we do not need to impose normalizations on $F_t$ that are common in the interactive fixed effects literature.

Because, in the post-treatment periods $t=0,\ldots,T_1$, $Y_{it}(0)$ is observed for individuals in the untreated group, \eqref{eq:2.3} suggests that we can recover $F^*_t$ using the untreated group if $\bm{Y}_{i, pre}$ is uncorrelated with $\bm \epsilon_{i,pre}$.
However, because they are correlated by construction, $F^*_t$ is not identified by simply regressing $Y_{it}(0)$ on $\bm{Y}_{i, pre}$ for individuals in the untreated group.

\begin{remark}
As mentioned above, \cite{brown2022unified} propose a general identification scheme for incorporating factor-model estimators for treatment effect estimation with a small number of periods.
Equation \eqref{eq:2.4} can be obtained by applying Theorem 1 of \cite{brown2022unified}.
However, our approach differs in how we use covariates to recover the key parameter $F^*_t$.
\end{remark}

\subsection{Main identification results}\label{3-sec:3.2}
In this section, we present our main identification strategy.
As a first step, in the spirit of Frisch-Waugh-Lovell Theorem, we rewrite \eqref{model-2} as
\begin{equation}\label{eq:3.1}
\xi_{it} = F_t' \eta_i + \epsilon_{it} \ \ \ \text{ for }t = -T_0, \ldots,T_1,
\end{equation}
where we define $\xi_{it} \equiv Y_{it}(0) - E[Y_{it} Z_i'|D_i=0] E[Z_i Z_i'|D_i=0]^{-1} Z_i$ and $\eta_i \equiv \lambda_i - E[\lambda_i Z_i'|D_i=0] E[Z_i Z_i'|D_i=0]^{-1} Z_i$.
From an argument analogous to Section \ref{3-sec:3.1}, we obtain an expression for $\xi_{it}$ in terms of their observed counterparts $\bm{\xi}_{i, pre} \equiv (\xi_{i,-1},\ldots,\xi_{i,-T_0})'$:
\begin{equation}\label{eq:a.1.3}
\xi_{it} = (F^{*}_t)' \bm \xi_{i,pre} + \epsilon_{it} - (F^{*}_t)' \bm \epsilon_{i,pre}.
\end{equation}
Similar to \eqref{eq:2.3}, \eqref{eq:a.1.3} follows from plugging in the expression of $\eta_i$ obtained from the pre-treatment periods to \eqref{eq:3.1}.
Then, by Assumptions \ref{asp:1} and \ref{asp:2}, 
\begin{equation}\label{eq:a.1.4}
E[\xi_{it}|D_i=1,Z_i=z] \ = \ (F^*_t)' E[\bm{\xi}_{i, pre}|D_i=1,Z_i=z]
\end{equation}
holds for any period $t\geq0$.
Therefore, if $F_{t}^*$ is identified, then $E[\xi_{it}|D_i=1,Z_i=z]$ and hence $E[Y_{it}(0)|D_i=1,Z_i=z]$ is also identified because all the other terms in the right-hand side of \eqref{eq:a.1.4} are immediately identified from the sampling process.

The central idea of our identification strategy is to employ correlation between $\xi_{it}$ and certain nonlinear transformations of $Z_i$ in the pre-treatment periods.
For $j=1,\ldots,R$, let $\omega_j:\mathcal{Z} \to \mathbb{R}$ be some nonlinear weight function.
Using the weight functions, we construct nonlinear transformations of $Z_i$.
We then construct a $R \times T_0$ matrix $\bm{\Omega}$ as follows:
\[
\bm{\Omega} \ \equiv \
\begin{pmatrix}
E[\omega_1(Z_i) \xi_{i,-1} |D_i=0] & \cdots & E[\omega_1(Z_i) \xi_{i,-T_0} |D_i=0]  \\
\vdots & & \vdots \\
E[\omega_R(Z_i) \xi_{i,-1} |D_i=0] & \cdots & E[\omega_R(Z_i) \xi_{i,-T_0} |D_i=0] 
\end{pmatrix}.
\]
This matrix is the correlation matrix between $\omega(Z_i) \equiv \left( \omega_1(Z_i), \ldots , \omega_R(Z_i) \right)'$ and $\bm{\xi}_{i, pre}$ for individuals in the untreated group.

We next introduce our main identifying assumption.
\begin{assumption}\label{asp:3}
A matrix $\bm{\Omega}$ has rank $r$.
\end{assumption}
\noindent Assumption \ref{asp:3} is the rank condition and requires that, for individuals in the untreated group, covariates with certain nonlinear transformations are sufficiently correlated with the outcome after eliminating the linear effects of covariates. 
This rank condition is the key assumption of the identification. 
Because $\bm{\Omega}$ is a $R \times T_0$ matrix, $T_0 \geq r$ and $R \geq r$ are the necessary conditions for Assumption \ref{asp:3} to hold.
Hence, Assumption \ref{asp:3} implies that we need at least $r$ pre-treatment periods and $r$ different weight functions.

Assumption \ref{asp:3} also requests that $\omega(Z_i)$ contains nonlinear functions of $Z_i$. If $\omega_j(Z_i) = c'Z_i$ for some $c \in \mathbb{R}^d$, then we have
\begin{eqnarray*}
E[\omega_j(Z_i) \xi_{it}|D_i=0] &=& c' E[Z_i \xi_{it}|D_i=0] \ = \ 0
\end{eqnarray*}
because $\xi_{it}$ is orthogonal to $Z_i$ by the definition of $\xi_{it}$. Hence, for satisfying Assumption \ref{asp:3}, we need to use nonlinear functions as the weight functions.
This requirement is not strong, and we could use, for example, orthogonal polynomials such as Hermite polynomials when $Z_i$ takes sufficiently many values.

An important implication of Assumption \ref{asp:3} is that, under Assumptions \ref{asp:1} and \ref{asp:2}, we could obtain a rank decomposition of $\bm{\Omega}$ with two meaningful matrices:
$$
\bm{\Omega} \ = \ \bm{H} \bm{F},
$$
where we define a $R \times r$ matrix $\bm{H}$ as 
\[
\bm{H} \equiv \left( E[ \omega_1(Z_i) \eta_i | D_i=0], \cdots, E[\omega_R(Z_i) \eta_i | D_i=0] \right)'.
\]
This matrix is the correlation matrix between $\omega(Z_i)$ and $\eta_i$ for individuals in the untreated group.
Because $\bm{H}$ is a $R \times r$ matrix and $\bm{F}$ is a $r \times T_0$ matrix, Assumption \ref{asp:3} implies not only that $\bm{F}$ is full row rank, but also that $\bm{H}$ is full column rank.
Hence, Assumption \ref{asp:3} requires that $\eta_i$, that contains elements of $\lambda_i$ that are nonlinear in $Z_i$, is sufficiently correlated with $\omega(Z_i)$.

To discuss the applicability of Assumption \ref{asp:3}, suppose that $\lambda_i$ is represented as $\lambda_i = h(Z_i) + u_i$, where $u_i$ is defined as $u_i\equiv\lambda_i - h(Z_i)$ and $h:\mathcal{Z} \to \mathbb{R}^r$ is a function of $Z_i$.
Without loss of generality, we assume that $h$ is a nonlinear function of $Z_i$ when $h(Z_i)\neq 0$ because linear terms of $Z_i$ are absorbed into the term $b_t' Z_i$ in the model \eqref{model-2}. 
If $h(Z_i)\neq 0$ and $h$ is nonlinear, Assumption \ref{asp:3} then requires sufficient correlation between $\lambda_i$ and $\omega(Z_i)$.
When $h(Z_i)=0$ and $u_i$ is exogenous, then there is no endogeneity because the terms $F_t' \lambda_i$ and $\epsilon_{it}$ are exogenous in the model \eqref{model-2}, and we could use existing methods such as pooled ordinary least squares for estimation.
When $h(Z_i)=0$ and $u_i$ are correlated with either $D_i$ or $Z_i$, Assumption \ref{asp:3} may hold when $u_i$ is sufficiently correlated with $Z_i$.

Next, we show that Assumption \ref{asp:3} provides the identification of $F_{t}^*$.
When we do not have access to pre-treatment periods, the available moment conditions relevant for $F_t$ are 
\begin{eqnarray}
0 &=& E[\omega_j(Z_i)\epsilon_{it} |D_i = 0] \nonumber \\
&=& E[\omega_j(Z_i) (\xi_{it} - F_t' \eta_i)|D_i=0] \text{ for }t=0,\ldots,T_1. \label{eq:2.5}
\end{eqnarray}
Suppose that we observe $\eta_i$ as well as $Z_i$.
Then, the moment conditions in \eqref{eq:2.5} suggest that we could recover $F_t$ by using $\omega(Z_i)$ as the instruments for $\eta_i$.
However, this estimation approach is infeasible because $\eta_i$ is unobservable.
Although $F_t$ itself is not identified, we only need to identify $F^*_t$, which is $F_t$ multiplied by $\bm F^+$, and we show that this parameter is identified using data on pre-treatment periods.

Notice that \eqref{eq:a.1.3} combined with Assumption \ref{asp:2} provides some moment conditions relevant for identifying $F^*_t$. 
The available moment conditions are
\begin{eqnarray}
0 &=& E[\omega_j(Z_i)(\epsilon_{it} - (F^{*}_t)' \bm \epsilon_{i,pre}) |D_i = 0] \nonumber \\
&=& E[\omega_j(Z_i)(\xi_{it} - (F^{*}_t)' \bm \xi_{i,pre}) |D_i = 0] \text{ for }t=0,\ldots,T_1. \label{eq:3.8}
\end{eqnarray}
The moment conditions in \eqref{eq:3.8} suggest that we could recover $F^*_t$ by using $\omega(Z_i)$ as the instruments for $\bm \xi_{i,pre}$.
Under Assumption \ref{asp:3}, we prove that $F^*_t$ is expressed as, for any $R \times R$ positive definite weighting matrix $\bm W$,
\begin{equation}\label{eq:f.ast}
F^*_t = (\bm W^{1/2} \bm \Omega)^+ \bm W^{1/2} \bm \Omega_t,
\end{equation}
where we define a $R\times 1$ vector $\bm{\Omega}_t$ as
\[
\bm{\Omega}_t \ \equiv \ \left( E[\omega_1(Z_i) \xi_{it} |D_i=0], \ldots, E[\omega_R(Z_i) \xi_{it} |D_i=0] \right)',
\]
and $\bm W^{1/2}$ is a $R \times R$ positive definite matrix that satisfies $\bm W=\bm W^{1/2}\bm W^{1/2}$.
Specifically, when $T_0=r$, the above expression of $F^*_t$ in \eqref{eq:f.ast} is
\begin{equation*}
F^*_t = (\bm \Omega' \bm W \bm \Omega)^{-1} \bm \Omega' \bm W \bm \Omega_t.
\end{equation*}
This expression is the population counterpart of the generalized method of moments (GMM) estimator. 
Although $\bm \Omega$ may not have full rank, $F^*_t$ has an expression analogous to the GMM estimator for $\bm{\xi}_{i, pre}$ instead of $\eta_i$ using $\omega(Z_i)$ as the instruments.
Because the sample counterpart of $\bm \Omega$ is observed from the pre-treatment periods, \eqref{eq:f.ast} shows that $F^*_t$ is identified.

We then state our main identification result.

\begin{theorem}\label{thm:1}
Under Assumptions \ref{asp:1}-\ref{asp:3}, $E[Y_{it}(0)|D_i=1, Z_i = z]$ is identified as follows for all $t \in \{0, \ldots, T_1\}$ and $z \in \mathcal{Z}$ from the distribution of $(\{Y_{it}\}_{t=-T_0}^{0},D_i,Z_i)$. For $t \in \{0, \ldots, T_1\}$, we have
\begin{equation}\label{eq:y.it}
E[Y_{it}(0)|D_i=1,Z_i=z] \ = \ 
\bm{\Omega}_t'\bm W^{1/2} (\bm{\Omega}' \bm W^{1/2})^{+} (E[\bm{Y}_{i,pre}|D_i=1,Z_i=z] - \bm{\beta}' z) + \beta_t' z,
\end{equation}
where $\beta_t \equiv E[Z_iZ_i'|D_i=0]^{-1} E[Z_i Y_{it}|D_i=0]$ and $\bm{\beta} \equiv (\beta_{-1}, \ldots, \beta_{-T_0})$. 
\end{theorem}
Theorem \ref{thm:1} establishes identification of $E[Y_{it}(0)|D_i=1, Z_i = z]$ for all $t \in \{0, \ldots, T_1\}$ under Assumptions \ref{asp:1}-\ref{asp:3}.

\section{Estimation}\label{3-sec:4}
In this section, we propose an estimation procedure based on our identification result.
Specifically, we consider that the treatment effect of a policy change affects only one unit.
This case is the main consideration of SC methods and is common in many applications of comparative case studies.
We consider the following setting with $N_0=N-1$:
\begin{equation*}
D_i \ = \ \begin{cases}
0, & i=1, \cdots, N_0 \\
1, & i=0.
\end{cases}
\end{equation*}
We propose a two-step procedure for estimation. 
In the first step, we estimate all the parameters $\beta_t$, $\bm{\Omega}_t$, and $\bm{\Omega}$.
Then, we plug these into the expression for ATT in Theorem \ref{thm:1}.

The expression for ATT in Theorem \ref{thm:1} suggests the following estimator of $Y_{0,t}(0)$,
\[
\hat{Y}_{0,t}(0) \ \equiv \ (\hat F_t^*)' (\bm{Y}_{0,pre} - \hat{\bm{\beta}}' Z_0) + \hat{\beta}_t' Z_0,
\]
where, for some $R \times R$ matrix $\hat{\bm W}^{1/2}$ that converges in probability to $\bm W^{1/2}$, $\hat F_t^*$ is defined as
\[
\hat F_t^* \equiv (\hat{\bm W}^{1/2} \hat{\bm{\Omega}})^+ \hat{\bm W}^{1/2} \hat{\bm{\Omega}}_t,
\]
$\hat{\bm{\beta}} \equiv (\hat{\beta}_{-1}, \ldots, \hat{\beta}_{-T_0})$, $\hat{\beta}_t$, $\hat{\bm{\Omega}}_t$, and $\hat{\bm{\Omega}}$ are the sample counterparts of $\bm{\Omega}_t$ and $\bm{\Omega}$, respectively.
$\hat{\beta}_t$ is defined as
\[
\hat{\beta}_t \ \equiv \ \left( \frac{1}{N_0} \sum_{i=1}^{N_0} Z_i Z_i' \right)^{-1} \frac{1}{N_0} \sum_{i=1}^{N_0} Z_i Y_{it} \ \ \text{for $t=-T_0, \ldots , 0$.}
\]
The $(j,t)$ element of $\hat{\bm{\Omega}}$ is defined as
\[
\hat{\Omega}_{j,-t} \equiv \frac{1}{N_0} \sum_{i=1}^{N_0} \omega_j(Z_i) \left( Y_{i,-t} - \hat{\beta}_{-t}'Z_i \right)
\]
and $\hat{\bm{\Omega}}_t$ is defined as $\hat{\bm{\Omega}}_t \equiv (\hat{\Omega}_{1,t},\ldots,\hat{\Omega}_{R,t})'$.

We make the following assumption:

\begin{assumption}\label{asp:4}
The data $(\{Y_{it}(0), Y_{it}(1)\}_{t=-T_0}^{T_1},D_i,Z_i)_{i=1}^{N_0}$ is identically distributed over $i$, and we have $E[||Y_{it}||_2^4|D_i=0]<\infty$
$E[||Z_i||_2^4|D_i=0]<\infty$, and $E[||\omega_j(Z_i)||_2^4|D_i=0]<\infty$.
\end{assumption}

\noindent Assumption \ref{asp:4} is a standard condition for establishing the probability limit by the law of large numbers (LLN).
Under Assumption \ref{asp:4}, $\hat{\bm{\Omega}}_t$, $\hat{\bm{\Omega}}$, and $\hat{\beta}_t$ converge in probability to the population counterparts.

Notice that the above estimator may not be valid when $\bm{\Omega}$ does not have full rank because the Moore-Penrose generalized inverse of a matrix is not necessarily a continuous function of the elements of the matrix.
Hence, $\hat{\bm{\Omega}}^+$ may not converge in probability to $\bm{\Omega}^+$.
\cite{stewart1969continuity} gives a necessary and sufficient condition for the continuity of the Moore-Penrose generalized inverse, and the continuity holds if and only if $\rank(\hat{\bm{\Omega}})=\rank(\bm{\Omega})$ holds for all sufficiently large $N_0$.
However, this condition may not be satisfied in our settings, and we suggest using the above estimator when $\bm{\Omega}$ has full rank, i.e., either $R=r$ or $T_0=r$ holds. 

The following result shows that, under $R=r$ or $T_0=r$, $\hat{Y}_{0,t}(0)$ is an asymptotically unbiased estimator of $Y_{0,t}(0)$ when $N_0\to\infty$.
\begin{theorem}\label{thm:2}
Suppose that either $R=r$ or $T_0=r$ holds.
Then, under Assumptions \ref{asp:1}-\ref{asp:4}, 
\begin{equation*}
\hat{Y}_{0,t}(0) = Y_{0,t}(0) - \epsilon_{0,t} + (F_t^*)' \bm{\epsilon}_{0,pre} + O_p\left(\frac{1}{\sqrt{N_0}}\right).
\end{equation*}
\end{theorem}

In general, we do not know the number of factors $r$, and we are likely to set $R$ and $T_0$ large enough for $R\geq r$ and $T_0\geq r$ to hold.
When $R> r$ and $T_0> r$, we propose the following alternative estimator
\[
\tilde{Y}_{0,t}(0) \ \equiv \ (\tilde F_t^*)' (\bm{Y}_{0,pre} - \hat{\bm{\beta}}' Z_0) + \hat{\beta}_t' Z_0,
\]
where $\tilde F_t^*$ is defined as
\[
\tilde F_t^* \equiv
(\hat{\bm{\Omega}}' \hat{\bm W} \hat{\bm{\Omega}} + \delta_{N_0} \bm I_{T_0} )^{-1} \hat{\bm{\Omega}}' \hat{\bm W} \hat{\bm{\Omega}}_t.
\]
We set $\delta_{N_0}>0$ as a tuning parameter that converges to 0 when $N_0 \to \infty$.
The above estimator is motivated by a well-known result (e.g., see \cite{ben2003generalized}) that for any matrix $\bm B$, we have $(\bm B' \bm B + \delta \bm I)^{-1} \bm B' \to \bm B^+$ when $\delta \to 0$.
Instead of using $(\bm W^{1/2} \hat{\bm{\Omega}}')^+$, we use $(\hat{\bm{\Omega}}' \bm W \hat{\bm{\Omega}} + \delta_{N_0} \bm I_{T_0} )^{-1} \hat{\bm{\Omega}}' \bm W^{1/2}$ for an estimator of $(\bm W^{1/2} \bm{\Omega}')^+$.

The following result shows that, when the tuning parameter $\delta_{N_0}$ converges to 0 at an appropriate rate, $\tilde{Y}_{0,t}(0)$ is an asymptotically unbiased estimator of $Y_{0,t}(0)$ when $N_0\to\infty$.

\begin{theorem}\label{thm:3}
Under Assumptions \ref{asp:1}-\ref{asp:4}, 
\begin{equation*}
\tilde{Y}_{0,t}(0) = Y_{0,t}(0) - \epsilon_{0,t} + (F_t^*)' \bm{\epsilon}_{0,pre} + 
O_p\left(\frac{1}{\sqrt{N_0}} + \delta_{N_0} + \frac{1}{\delta_{N_0} N_0}\right).
\end{equation*}
\end{theorem}

\noindent Theorem \ref{thm:3} states that $\tilde{Y}_{0,t}(0)$ is a valid estimator when the convergence of the tuning parameter $\delta_{N_0}$ is slower than $1/N_0$.

We summarize some interesting features of our estimation approach in the remarks below.

\begin{remark}
Although our estimator is constructed to have good properties even when $T_0$ is small, the bias of our estimator is expected to be smaller when $T_0$ is larger.
Observe that the remained bias of our estimator when $N_0$ is sufficiently large is
\[
-\epsilon_{0,t} + F_t' (\bm{F} \bm{F}')^{-1}\bm{F} \bm{\epsilon}_{0,pre}
= -\epsilon_{0,t} + F_t' \left( \frac{1}{T_0} \sum_{s=1}^{T_0} F_{-s} F_{-s}' \right)^{-1} \left( \frac{1}{T_0} \sum_{s=1}^{T_0} F_{-s} \epsilon_{0,-s} \right).
\]
Hence, when $\frac{1}{T_0} \sum_{s=1}^{T_0} F_{-s} F_{-s}'$ converges to a positive definite matrix and $\sum_{s=1}^{T_0} F_{-s} \epsilon_{N,-s}$ converges in probability to 0 when $T_0 \to \infty$, the bias term is expected to be as small as the error term $\epsilon_{0,t}$ when $T_0$ as well as $N_0$ is sufficiently large.
\end{remark}

\begin{remark}
Our alternative estimator $\tilde F_t^*$ of $F_t^*$ is a unique minimizer of the following minimization problem:
\begin{equation*}
\min_{f\in \mathbb{R}^{T_0}} (\hat{\bm{\Omega}}_t - \hat{\bm{\Omega}} f)' \hat{\bm W} (\hat{\bm{\Omega}}_t - \hat{\bm{\Omega}} f ) + \delta_{N_0} ||f||_2^2.
\end{equation*}
This problem is a type of well known Tikhonov regularization (e.g., see \cite{ben2003generalized}), and our estimator of $F_t^*$ is related to ridge regression (\cite{hoerl1970ridge}).
For linear factor models with a small number of periods, \cite{imbens2021controlling} also propose a similar type of regularized estimator in their context and obtain the same rate of convergence (they use the whole sample size) as in Theorem \ref{thm:3} for the tuning parameter.
\end{remark}







\section{Simulation study}\label{3-sec:5}
In this section, we provide Monte Carlo simulations to illustrate the finite sample properties of our estimation procedure. 
In particular, we compare the performance of our approach to DID and SC methods.
Suppose that the potential outcomes are generated as follows:
\[
Y_{it}(0) = b_{1t} + b_{2t} Z_i + F_{1t} \lambda_{1i} + F_{2t} \lambda_{2i} + \epsilon_{it}, \ \ \ t = -T_0, \ldots,0.
\]
We consider the case where $T_0=5,10$ and $T_1=0$ for the number of periods and $N=40,100$ for the number of units.
We set $Y_{00} = 1 + Y_{00}(0)$, and the effect of participating in the treatment is 1.
For the observed covariates, we set $Z_i\sim N(0,1)$ for $i=1,\ldots,N_0$ and $Z_0\sim N(1,1)$ for the treated unit.
For the effects of these covariates, we set
\[
\begin{cases}
b_{1t} = \{(t - 1)/T_0\} + 1 & \\
b_{2t} = \{(t - 1)/T_0\}^2 + 1. &
\end{cases}
\]
For the unobserved factors, we set
\[
\begin{cases}
\lambda_{1i} = \log(1 + Z_i^4) - E[\log(1 + Z_1^4)] + u_{1i} & \\
\lambda_{2i} = 0.5 \{ \exp(-0.2 Z_i) - E[\exp(-0.2 Z_1)] \} + u_{2i}, &
\end{cases}
\]
where we set $u_{1i},u_{2i}\sim N(0,0.04)$.
We generate $2(T_0+1)$ values from $N(0,1)$ and set them as the (fixed) values of $F_{1t}$ and $F_{2t}$.
Finally, we set $\epsilon_{it}\sim N(0,1)$, and $(Z_i,\epsilon_{it},u_{1i},u_{2i})$ are independently generated.

Under this setting, DID is not correctly specified because of multiple factors. 
Similarly, the methods of \cite{imbens2021controlling} and \cite{callaway2022treatment} cannot be applied because we have only one post-treatment period, and the effects of covariates are not time invariant.

For our methods, we compare two cases $R=2$ and $3$.
The weighting functions are Hermite polynomials, that are $\omega_1(z)=4z^{2}-2$ and $\omega_2(z)=8z^{3}-12z$, and $\omega_3(z)=16z^4 - 48z^2 + 12$, normalized with sample means and sample variances.
In this setting, the number of factors is $r=2$.
For the $R=2$ case, we consider two types of estimators.
The first one directly uses the Moore-Penrose generalized inverse of the sample matrix for estimation.
This can be applied when the number of factors is correctly specified.
The second one is the alternative estimator with a tuning parameter $\delta_{N_0}$.
We select $\delta_{N_0}$ using (delete-one) cross-validation or generalized cross-validation.
For the $R=3$ case, we employ the alternative estimator with a tuning parameter $\delta_{N_0}$.
For the SC methods, we compare two types of predictors; predictors I: All pre-treatment outcome lags, and predictors II: Half of the pre-treatment outcome lags and covariates.
These specifications are based on the suggestions of \cite{ferman2020cherry}.

\begin{table}
\caption{Monte Carlo Simulations ($T_0=5$)}\label{tab:5.1}
\vspace{0.2cm}
\centering
\scalebox{0.83}{
\begin{tabular}{c|c c|c c|c c|c c|c c} \hline \hline
\multicolumn{1}{c|}{} & \multicolumn{2}{|c|}{$R=2$ (no tuning)} &  \multicolumn{2}{|c}{$R=2$ (CV)} & \multicolumn{2}{|c}{$R=2$ (GCV)} &  \multicolumn{2}{|c}{$R=3$ (CV)} &  \multicolumn{2}{|c}{$R=3$ (GCV)} \\
$N$ & 40 & 100 & 40 & 100 & 40 & 100 & 40 & 100 & 40 & 100   \\ \hline
bias & -0.059 & 0.123 & -0.435 & -0.285 & -0.574 & -0.387 & -0.431 & -0.276 & -0.575 & -0.389  \\ 
sd & 2.025 & 1.969 & 1.665 & 1.588 & 1.672 & 1.633 & 1.676 & 1.586 & 1.675 & 1.632 \\ 
RMSE & 2.026 & 1.973 & 1.720 & 1.613 & 1.768 & 1.678 & 1.731 & 1.610 & 1.771 & 1.678 \\  \hline \hline
\end{tabular}
}
\vspace{0.5cm}

\caption{Monte Carlo Simulations ($T_0=5$)}\label{tab:5.2}
\vspace{0.2cm}
\centering
\scalebox{0.83}{
\begin{tabular}{c|c c|c c|c c} \hline \hline
\multicolumn{1}{c|}{} & \multicolumn{2}{|c}{DID} &  \multicolumn{2}{|c}{SCM-I} &  \multicolumn{2}{|c}{SCM-II} \\
$N$ & 40 & 100 & 40 & 100 & 40 & 100  \\ \hline
bias & -1.147 & -1.031 & -0.577 & -0.146 & -0.355 & 0.228  \\ 
sd & 2.471 & 2.470 & 2.239 & 1.874 & 2.213 & 1.732  \\ 
RMSE & 2.724 & 2.676 & 2.313 & 1.880 & 2.241 & 1.747  \\  \hline \hline
\end{tabular}
}
\vspace{0.5cm}

\caption{Monte Carlo Simulations ($T_0=10$)}\label{tab:5.3}
\vspace{0.2cm}
\centering
\scalebox{0.83}{
\begin{tabular}{c|c c|c c|c c|c c|c c} \hline \hline
\multicolumn{1}{c|}{} & \multicolumn{2}{|c|}{$R=2$ (no tuning)} &  \multicolumn{2}{|c}{$R=2$ (CV)} & \multicolumn{2}{|c}{$R=2$ (GCV)} &  \multicolumn{2}{|c}{$R=3$ (CV)} &  \multicolumn{2}{|c}{$R=3$ (GCV)} \\
$N$ & 40 & 100 & 40 & 100 & 40 & 100 & 40 & 100 & 40 & 100   \\ \hline
bias & -0.113 & 0.045 & -0.217 & -0.041 & -0.261 & -0.067 & -0.220 & -0.038 & -0.267 & -0.066 \\ 
sd & 1.299 & 1.229 & 1.192 & 1.159 & 1.199 & 1.147 & 1.193 & 1.155 & 1.204 & 1.144 \\ 
RMSE & 1.304 & 1.230 & 1.211 & 1.159 & 1.227 & 1.149 & 1.213 & 1.156 & 1.234 & 1.146 \\  \hline \hline
\end{tabular}
}
\vspace{0.5cm}

\caption{Monte Carlo Simulations ($T_0=10$)}\label{tab:5.4}
\vspace{0.2cm}
\centering
\scalebox{0.83}{
\begin{tabular}{c|c c|c c|c c} \hline \hline
\multicolumn{1}{c|}{} & \multicolumn{2}{|c}{DID} &  \multicolumn{2}{|c}{SCM-I} &  \multicolumn{2}{|c}{SCM-II} \\
$N$ & 40 & 100 & 40 & 100 & 40 & 100  \\ \hline
bias & -0.508 & -0.388 & -0.461 & -0.150 & -0.385 & 0.021 \\ 
sd & 1.490 & 1.467 & 1.634 & 1.361 & 1.674 & 1.406 \\ 
RMSE & 1.575 & 1.517 & 1.697 & 1.369 & 1.717 & 1.406 \\  \hline \hline
\end{tabular}
}
\vspace{0.5cm}

\begin{minipage}{435pt}
{\fontsize{10pt}{10pt}\selectfont\smallskip\textit{Notes}: The columns labeled ‘‘$R=2$’’ and ‘‘$R=3$’’ use the estimation approach introduced in the paper. 
For the estimation approach introduced in the paper, we set the weighting matrix to be an identity matrix.
The columns labeled with ‘‘(no tuning)’’ corresponds to directly using the Moore-Penrose generalized inverse of the sample matrix for estimation.
The columns labeled with ‘‘(CV)’’ and ‘‘(GCV)’’ corresponds to the alternative estimators where the tuning parameter $\delta_{N_0}$ is selected by (delete-one) cross-validation (CV) and generalized cross-validation (GCV), respectively.
The columns labeled ‘‘DID’’ provide results using DID approach; the columns labeled ‘‘SCM-I’’ and ‘‘SCM-II’’ provide results using SC methods with predictors I and II, respectively.
The columns labeled ‘‘bias’’ simulate the bias, the columns labeled  ‘‘sd’’ simulates the standard deviation, and the columns labeled ‘‘RMSE’’ simulates the root mean squared error of each approach.
We conduct 500 simulations after eliminating the samples where an error occurred in the optimization of the SC methods.
}
\end{minipage}

\end{table}

Tables \ref{tab:5.1}--\ref{tab:5.4} contain the results of this experiment for $T_0 = 5$ and $10$.
The number of replications is set at 500 throughout.
Tables \ref{tab:5.1} and \ref{tab:5.3} show the bias, standard deviation, and root mean squared error of the proposed estimators.
Tables \ref{tab:5.2} and \ref{tab:5.4} show the bias, standard deviation, and root mean squared error of the DID and SC methods.
As $N$ and $T_0$ increase, the bias and standard deviations tend to decrease for the SC and the proposed estimators.
Because of the factor structure, the DID method does not perform well for all settings.
The SC method with predictors I has a larger bias than the SC method with predictor II.
The bias of the proposed estimator is smaller than other methods for most settings especially when the number of weighted functions is appropriately selected and one directly uses the Moore-Penrose generalized inverse of the sample matrix for estimation.
The proposed estimators have a smaller RMSE than the DID and the SC methods.



\section{Conclusion}\label{3-sec:6}
In this paper, we have developed a new approach for identifying and estimating the ATT when untreated potential outcomes have a linear factor structure.
We highlight some attractive features of our approach for use in applied work.
First, our approach generalizes the common approaches for policy evaluation, such as DID methods, to the cases where the parallel trends assumption may not hold due to the time-varying unobservables.
Second, our approach does not require the number of pre-treatment periods to be sufficiently large for valid estimation.
Our estimation approach is expected to have good properties against other methods on linear factor models when the researcher does not have access to enough periods of data to validate their large $T_0$ asymptotics.

The main requirement for using our approach is that the researcher has access to time invariant variables that are sufficiently correlated with the pre-treatment outcomes in that a certain nonlinear transformation of these variables has enough correlation with the pre-treatment outcomes.
A good candidate for this variable would be a variable that takes many values.
Even when there are only discrete variables with small supports, our approach may be applicable if the researcher has access to a variety of covariates and their whole support is large enough for certain nonlinear transformations.


\bibliography{reference} 

\appendix

\section{Proofs of the main results}\label{sec:a}
Appendix \ref{sec:a} provides proofs of Theorems \ref{thm:1}-\ref{thm:3}.
Proof of Theorems \ref{thm:2} and \ref{thm:3} use some preliminary or auxiliary results that are collected in Lemmas \ref{lem:p1} and \ref{lem:a1} in Appendices \ref{sec:b} and \ref{sec:c}.
For any matrix $\bm B$, let $||\bm B||_2$ denote the matrix 2-norm, that is, the largest singular value of $\bm B$.

\begin{proof}[Proof of Theorem \ref{thm:1}]
We first prove that \eqref{eq:a.1.3} and \eqref{eq:a.1.4} hold for any period $t\geq0$.
For any $t \geq -T_0$, we observe that, from the definition of $\xi_{it}$,
\begin{eqnarray}
\xi_{it} &=& Y_{it}(0) - E[Y_{it}(0) Z_i'|D_i=0] E[Z_i Z_i'|D_i=0]^{-1} Z_i  \nonumber \\
&=& b_t'Z_i + F_t' \lambda_i + \epsilon_{it} - b_t' Z_i - F_t' E[\lambda_i Z_i'|D_i=0] E[Z_i Z_i'|D_i=0]^{-1} Z_i \nonumber \\
&=& F_t' \eta_{i} + \epsilon_{it}. \label{a.model-2}
\end{eqnarray}
Then, we have
\begin{eqnarray}\label{eq:a.1.1}
\bm{\xi}_{i,pre} &=& \begin{pmatrix}
F_{-1}' \eta_{i} + \epsilon_{i, -1} \\
\vdots \\
F_{-T_0}' \eta_{i} + \epsilon_{i, -T_0}
\end{pmatrix} \ = \ \bm{F}' \eta_i + \bm{\epsilon}_{i,pre}.
\end{eqnarray}
Because $T_0\geq r$ and $\bm{F}$ is full row rank by Assumption \ref{asp:3}, we can obatin an expression for $\eta_i$ by re-arranging the terms in \eqref{eq:a.1.1}:
\begin{equation}\label{eq:a.1.2}
\eta_i = (\bm F')^{+}(\bm \xi_{i,pre} - \bm \epsilon_{i,pre}).
\end{equation}
By plugging in the expression for $\eta_i$ in \eqref{eq:a.1.2} to \eqref{a.model-2} and re-arranging terms, it follows that \eqref{eq:a.1.3} holds for any period $t\geq0$.
Then, by Assumptions \ref{asp:1} and \ref{asp:2}, \eqref{eq:a.1.4} holds for any period $t\geq0$.

Next we prove that $F^*_t$ can be expressed as \eqref{eq:f.ast}.
From \eqref{a.model-2}, we have
\begin{eqnarray}
\bm{\Omega} &=& \begin{pmatrix}
E[ \omega_1(Z_i) \eta_i' | D_i=0] F_{-1} & \cdots & E[ \omega_1(Z_i) \eta_i' | D_i=0] F_{-T_0} \nonumber \\
\vdots &  & \vdots \\
E[ \omega_R(Z_i) \eta_i' | D_i=0] F_{-1} & \cdots & E[ \omega_R(Z_i) \eta_i' | D_i=0] F_{-T_0}
\end{pmatrix} \nonumber \\
&=& \begin{pmatrix}
E[ \omega_1(Z_i) \eta_i' | D_i=0] \\
\vdots \\
E[ \omega_R(Z_i) \eta_i' | D_i=0]
\end{pmatrix}
\left( F_{-1}, \cdots, F_{-T_0} \right) \ = \ \bm{H} \bm{F} \label{eq:Omega},
\end{eqnarray}
where the first equality follows from Assumption \ref{asp:2}.
Similarly, we obtain
\begin{eqnarray}
\bm{\Omega}_t &=&
\begin{pmatrix}
E[ \omega_1(Z_i) \eta_i' | D_i=0] F_t  \\
\vdots \\
E[ \omega_R(Z_i) \eta_i' | D_i=0] F_t
\end{pmatrix}  =  \bm{H} F_t, \label{eq:Omega.t}
\end{eqnarray}
and multiplying both sides of \eqref{eq:Omega.t} by $\bm H' \bm W$ gives
\begin{equation*}
\bm H' \bm W \bm{\Omega}_t = \bm H' \bm W \bm{H} F_t.
\end{equation*}
Because $\bm H' \bm W \bm H$ is a $r\times r$ invertible matrix, we obtain the following expression for $F_t$ by re-arranging the terms.
\begin{equation}\label{eq:f}
F_t = (\bm H' \bm W \bm H)^{-1}\bm H' \bm W \bm \Omega_t 
= (\bm W^{1/2} \bm H)^+ \bm W^{1/2} \bm \Omega_t,
\end{equation}
Observe that, by multiplying both sides of \eqref{eq:f} by $\bm F^+$, we have
\begin{equation}\label{eq:a.1.7}
F^*_t = \bm F^+ (\bm W^{1/2} \bm H)^+ \bm W^{1/2} \bm \Omega_t.
\end{equation}
Notice that, because $\bm{F}$ is full row rank and $\bm W^{1/2} \bm H$ is full column rank from Assumption \ref{asp:3}, it follows that
\begin{equation}\label{eq:a.1.8}
\bm F^+ (\bm W^{1/2} \bm H)^+ = (\bm W^{1/2} \bm H \bm F)^+.
\end{equation}
Then, plugging in \eqref{eq:a.1.8} to the expression of $F^*_t$ in \eqref{eq:a.1.7} gives
\begin{equation}\label{eq:a.1.9}
F^*_t = (\bm W^{1/2} \bm H \bm F)^+ \bm W^{1/2} \bm \Omega_t 
= (\bm W^{1/2} \bm \Omega)^+ \bm W^{1/2} \bm \Omega_t,
\end{equation}
where the second equality follows from \eqref{eq:Omega}.
Therefore, we obtain an expression for $F^*_t$ in \eqref{eq:f.ast}.

Finally, we prove that $E[Y_{it}(0)|D_i=1,Z_i=z]$ is expressed as in \eqref{eq:y.it}.
We observe that, from the definition of $\xi_{it}$,
\begin{equation}\label{eq:a.1.5}
E[Y_{it}(0)|D_i=1,Z_i=z] = E[\xi_{it}|D_i=1,Z_i=z] + \beta_t' z  
\end{equation}
holds for any period $t\geq-T_0$.
Then, from \eqref{eq:a.1.4} and \eqref{eq:a.1.5}, we have
\begin{eqnarray}
E[Y_{it}(0)|D_i=1,Z_i=z] &=& (F^*_t)' E[\bm{\xi}_{i, pre}|D_i=1,Z_i=z] + \beta_t' z \nonumber \\
&=& (F^*_t)' (E[\bm{Y}_{i,pre}|D_i=1,Z_i=z] - \bm \beta' z) + \beta_t' z \label{eq:a.1.6}
\end{eqnarray}
for any period $t\geq0$.
Hence, we obtain \eqref{eq:y.it} by plugging in an expression for $F^*_t$ in \eqref{eq:f.ast} to \eqref{eq:a.1.6}.
Therefore, $E[Y_{it}(0)|D_i=1, Z_i = z]$ is identified for all $t \in \{0, \ldots, T_1\}$ and $z \in \mathcal{Z}$ from the distribution of $(\{Y_{it}\}_{t=-T_0}^{0},D_i,Z_i)$.
\end{proof}

For notation simplicity, we let $\bm A \equiv \bm W^{1/2} \bm{\Omega}$, $\hat{\bm A} \equiv \hat{\bm W}^{1/2} \hat{\bm{\Omega}}$, $\bm a_t \equiv \bm W^{1/2} \bm{\Omega}_t$, and $\hat{\bm a}_t \equiv \hat{\bm W}^{1/2} \hat{\bm{\Omega}}_t$.
Then, from the definition of $\hat F_t^*$ and $\tilde F_t^*$ and the expression of $F_t^*$ in \eqref{eq:f.ast}, $\hat F_t^*$, $\tilde F_t^*$ and $F_t^*$ are expressed as
\[
\hat F_t^* = (\hat{\bm A}^+)'\hat{\bm a}_t ,\quad
\tilde F_t^* =
(\hat{\bm A}' \hat{\bm A} + \delta_{N_0} \bm I_{T_0} )^{-1} \hat{\bm A}' \hat{\bm a}_t \quad\text{and}\quad
F_t^* = \bm A^+\bm a_t.
\]

\begin{proof}[Proof of Theorem \ref{thm:2}]
First, we prove that 
\begin{equation}\label{eq:a.2.1}
||\tilde F_t^* - F_t^*||_2= O_p\left(\frac{1}{\sqrt{N_0}}\right).
\end{equation}
Observe that, $\bm A^+$ is expressed as $\bm A^+ = (\bm A'\bm A)^{-1}\bm A'$
when $T_0=r$ and hence $\bm \Omega$ has full column rank, and $(\bm W^{1/2} {\bm{\Omega}})^+$ is expressed as $\bm A^+ = \bm A'( \bm A \bm A')^{-1}$
when $R=r$ and hence $\bm \Omega$ has full row rank.
When $T_0=r$, observe that
\begin{equation*}
\hat F_t^* = (\hat{\bm A}' \hat{\bm A})^{-1}\hat{\bm A}'
(\hat{\bm a}_t - \hat{\bm A} F_t^*) + F_t^*
\end{equation*}
and hence we obtain
\begin{equation*}
||\tilde F_t^* - F_t^*||_2
\leq ||(\hat{\bm A}' \hat{\bm A})^{-1}\hat{\bm A}'||_2
||\hat{\bm a}_t - \hat{\bm A} F_t^* ||_2
= O_p\left(\frac{1}{\sqrt{N_0}}\right).
\end{equation*}
When $R=r$, observe that
\begin{eqnarray*}
\tilde F_t^* - F_t^*
= \{\hat{\bm A}'(\hat{\bm A}\hat{\bm A}')^{-1}
- \bm A'( \bm A \bm A')^{-1}\}\hat{\bm a}_t
+ \bm A'( \bm A \bm A')^{-1}(\hat{\bm a}_t - \bm a_t).
\end{eqnarray*}
For the first term, observe that
\begin{eqnarray*}
&& \hat{\bm A}'(\hat{\bm A}\hat{\bm A}')^{-1}
- \bm A'( \bm A \bm A')^{-1} \\
&=& \hat{\bm A}'\{(\hat{\bm A}\hat{\bm A}')^{-1}
-( \bm A \bm A')^{-1}\}
+ (\hat{\bm A} - \bm A)'( \bm A \bm A')^{-1} \\
&=& \hat{\bm A}'(\hat{\bm A}\hat{\bm A}')^{-1}
(\bm A \bm A'-\hat{\bm A}\hat{\bm A}')( \bm A \bm A')^{-1}
+ (\hat{\bm A} - \bm A)'( \bm A \bm A')^{-1} \\
&=& \hat{\bm A}'(\hat{\bm A}\hat{\bm A}')^{-1}\hat{\bm A}(\bm A - \hat{\bm A})'( \bm A \bm A')^{-1} \\
&& + \hat{\bm A}'(\hat{\bm A}\hat{\bm A}')^{-1}(\bm A - \hat{\bm A})\hat{\bm A}'( \bm A \bm A')^{-1} + (\hat{\bm A} - \bm A)'( \bm A \bm A')^{-1},
\end{eqnarray*}
and hence we obtain
\begin{eqnarray*}
||\{\hat{\bm A}'(\hat{\bm A}\hat{\bm A}')^{-1}
- \bm A'( \bm A \bm A')^{-1}\}\hat{\bm a}_t ||_2
= O_p\left(\frac{1}{\sqrt{N_0}}\right).
\end{eqnarray*}
From Lemma \ref{lem:p1} (a), we have for the second term that
\begin{eqnarray*}
||\bm A'( \bm A \bm A')^{-1}(\hat{\bm a}_t - \bm a_t) ||_2
= O_p\left(\frac{1}{\sqrt{N_0}}\right),
\end{eqnarray*}
and hence we obtain \eqref{eq:a.2.1}.

Next, observe that
\begin{eqnarray*}
&& |\hat{Y}_{0,t}(0) - (F_t^*)' (\bm{Y}_{0,pre} - \bm{\beta}' Z_0) + \beta_t' Z_0| \\
&\leq & |(\hat F_t^* - F_t^*)'(\bm{Y}_{0,pre} - \hat{\bm{\beta}}' Z_0) |
+|(F_t^*)'(\bm{\beta} - \hat{\bm{\beta}} )'Z_0 |
+ |(\hat \beta_t - \beta_t)'Z_0 | \\
&\leq & ||\hat F_t^* - F_t^*||_2||\bm{Y}_{0,pre} - \hat{\bm{\beta}}' Z_0 ||_2
+||F_t^*||_2||\bm{\beta} - \hat{\bm{\beta}}||_2||Z_0||_2
+||\hat \beta_t - \beta_t||_2||Z_0 ||_2 \\
&=& O_p\left(\frac{1}{\sqrt{N_0}}\right),
\end{eqnarray*}
and hence we obtain
\begin{equation*}
\hat{Y}_{0,t}(0) = (F_t^*)' (\bm{Y}_{0,pre} - \bm{\beta}' Z_0) + \beta_t' Z_0
=O_p\left(\frac{1}{\sqrt{N_0}}\right)
\end{equation*}

Finally, we prove that
\begin{eqnarray}\label{eq:a.2.14}
 (F_t^*)' (\bm{Y}_{0,pre} - \bm{\beta}' Z_0) + \beta_t' Z_0 
= Y_{0,t}(0) - \epsilon_{0,t} + (F_t^*)'\bm{\epsilon}_{0,pre}. \label{eq:a.2.4}
\end{eqnarray}
We observe that, from the definition of $\xi_{it}$,
\begin{equation}\label{eq:a.2.5}
Y_{0t}(0) = \xi_{it} + \beta_t' Z_0  
\end{equation}
holds for any period $t\geq-T_0$.
Then, from \eqref{eq:a.1.4} and \eqref{eq:a.2.5}, we have
\begin{eqnarray}
Y_{0t}(0) &=& (F^*_t)' \bm{\xi}_{0, pre} + \beta_t' Z_0 + \epsilon_{0,t} - (F_t^*)'\bm{\epsilon}_{0,pre} \nonumber \\
&=& (F^*_t)' (\bm{Y}_{0,pre} - \bm \beta' Z_0) + \beta_t' Z_0 + \epsilon_{0,t} - (F_t^*)'\bm{\epsilon}_{0,pre} \label{eq:a.2.6}
\end{eqnarray}
for any period $t\geq0$.
Hence, we obtain \eqref{eq:a.2.4} by plugging in an expression for $F^*_t$ in \eqref{eq:f.ast} to \eqref{eq:a.2.6} and re-arranging terms.

\end{proof}

\begin{proof}[Proof of Theorem \ref{thm:3}]
We first show that
\begin{equation}\label{eq:a.3.17}
||\tilde F_t^* - F_t^*||_2 
= O_p\left(\delta_{N_0} + \frac{1}{\sqrt{N_0}} + \frac{1}{\delta_{N_0} N_0}\right).
\end{equation}
Observe that 
\begin{eqnarray*}
&& \tilde F_t^* - F_t^* \\
&=& (\hat{\bm A}' \hat{\bm A} + \delta_{N_0} \bm I_{T_0} )^{-1}\hat{\bm A}'(\hat{\bm a}_t - \hat{\bm A} F_t^* + \hat{\bm A} F_t^* ) \\
&& - \{\bm A^+ - (\bm A' \bm A + \delta_{N_0} \bm I_{T_0} )^{-1}\bm A'
+ (\bm A' \bm A + \delta_{N_0} \bm I_{T_0} )^{-1}\bm A'\}\bm a_t \\
&=& (\hat{\bm A}' \hat{\bm A} + \delta_{N_0} \bm I_{T_0} )^{-1}\hat{\bm A}'(\hat{\bm a}_t - \hat{\bm A} F_t^*) \\
&& + \{(\hat{\bm A}' \hat{\bm A} + \delta_{N_0} \bm I_{T_0} )^{-1}\hat{\bm A}'\hat{\bm A} - (\bm A' \bm A + \delta_{N_0} \bm I_{T_0} )^{-1}\bm A'\bm A\}F_t^* \\
&& + \{\bm A^+ - (\bm A' \bm A + \delta_{N_0} \bm I_{T_0} )^{-1}\bm A'\}\bm a_t.
\end{eqnarray*}
In the second equality, we use the first statement of Lemma \ref{lem:p1} (b).
From Lemma \ref{lem:a1} (b), we have
\[
||(\bm A' \bm A + \delta_{N_0} \bm I_{T_0} )^{-1}\bm A' ||_2 
= \frac{\sigma_r}{\sigma_r^2 + \delta_{N_0}} = O_p(1),
\]
and, from Lemma \ref{lem:p1} (b) and (c), we have for the first term that
\begin{eqnarray*}
&& ||(\hat{\bm A}' \hat{\bm A} + \delta_{N_0} \bm I_{T_0} )^{-1}\hat{\bm A}'(\hat{\bm a}_t - \hat{\bm A} F_t^*) ||_2 \\
&\leq& ||(\bm A' \bm A + \delta_{N_0} \bm I_{T_0} )^{-1}\bm A' ||_2
||\hat{\bm a}_t - \hat{\bm A} F_t^* ||_2 \\
&& + ||(\hat{\bm A}' \hat{\bm A} + \delta_{N_0} \bm I_{T_0} )^{-1}\hat{\bm A}' - (\bm A' \bm A + \delta_{N_0} \bm I_{T_0} )^{-1}\bm A'||_2
||\hat{\bm a}_t - \hat{\bm A} F_t^* ||_2 \\
&=& O_p\left(\frac{1}{\sqrt{N_0}} + \frac{1}{\delta_{N_0} N_0}\right).
\end{eqnarray*}
From Lemma \ref{lem:p1} (d), we have for the second term that
\begin{equation*}
||\{(\hat{\bm A}' \hat{\bm A} + \delta_{N_0} \bm I_{T_0} )^{-1}\hat{\bm A}'\hat{\bm A} - (\bm A' \bm A + \delta_{N_0} \bm I_{T_0} )^{-1}\bm A'\bm A\}F_t^* ||_2 = O_p\left(\frac{1}{\sqrt{\delta_{N_0} N_0}}\right).
\end{equation*}
From Lemma \ref{lem:a1} (d), we have
\[
||\bm A^+ - (\bm A' \bm A + \delta_{N_0} \bm I_{T_0} )^{-1}\bm A' ||_2 
= \frac{\delta_{N_0}}{\sigma_r(\sigma_r^2 + \delta_{N_0})}
=O_p(\delta_{N_0}),
\]
and we have, for the third term that
\begin{equation*}
||\{\bm A^+ - (\bm A' \bm A + \delta_{N_0} \bm I_{T_0} )^{-1}\bm A'\}\bm a_t ||_2 = O_p(\delta_{N_0}).
\end{equation*}
Therefore, we obtain \eqref{eq:a.3.17}.

We next show that
\begin{equation}\label{eq:a.3.22}
\tilde{Y}_{0,t}(0) = (F_t^*)' (\bm{Y}_{0,pre} - \bm{\beta}' Z_0) + \beta_t' Z_0
+O_p\left(\delta_{N_0} + \frac{1}{\sqrt{N_0}} + \frac{1}{\delta_{N_0} N_0}\right).
\end{equation}
Observe that 
\begin{eqnarray*}
&& |\tilde{Y}_{0,t}(0) - (F_t^*)' (\bm{Y}_{0,pre} - \bm{\beta}' Z_0) + \beta_t' Z_0| \\
&\leq & |(\tilde F_t^* - F_t^*)'(\bm{Y}_{0,pre} - \hat{\bm{\beta}}' Z_0) |
+|(F_t^*)'(\bm{\beta} - \hat{\bm{\beta}} )'Z_0 |
+ |(\hat \beta_t - \beta_t)'Z_0 | \\
&\leq & ||\tilde F_t^* - F_t^*||_2||\bm{Y}_{0,pre} - \hat{\bm{\beta}}' Z_0 ||_2
+||F_t^*||_2||\bm{\beta} - \hat{\bm{\beta}}||_2||Z_0||_2
+||\hat \beta_t - \beta_t||_2||Z_0 ||_2 \\
&=& O_p\left(\delta_{N_0} + \frac{1}{\sqrt{N_0}} + \frac{1}{\delta_{N_0} N_0}\right) 
+ O_p\left(\frac{1}{\sqrt{N_0}}\right) 
= O_p\left(\delta_{N_0} + \frac{1}{\sqrt{N_0}} + \frac{1}{\delta_{N_0} N_0}\right),
\end{eqnarray*}
and hence we obtain \eqref{eq:a.3.22}.
Therefore, the stated result follows from \eqref{eq:a.3.22} combined with \eqref{eq:a.2.14}.
\end{proof}

\section{Preliminary Lemma}\label{sec:b}
The following lemma provides the convergence rates of the important terms that appear in the proof of Theorems \ref{thm:2} and \ref{thm:3}.
Proofs in this section use some auxiliary results collected in Lemma \ref{lem:a1} in Appendix \ref{sec:c}.
We let $\rank(\hat{\bm A})=\hat r$, and we denote the $i$th largest singular values of $\bm A$ and $\hat{\bm A}$ as $\sigma_i$ and $\hat \sigma_i$, respectively.

\begin{lemma}\label{lem:p1}
Suppose Assumptions \ref{asp:1}-\ref{asp:4} hold.
\begin{enumerate}
\item[(a)] 
\[
||\hat{\bm A} - \bm A||_2 = O_p\left(\frac{1}{\sqrt{N_0}}\right)
\quad\text{and}\quad
||\hat{\bm a}_t - \bm a_t||_2 = O_p\left(\frac{1}{\sqrt{N_0}}\right).
\]
\item[(b)]
\[
\bm a_t = \bm A F_t^*
\quad\text{and}\quad
||\hat{\bm a}_t - \hat{\bm A}F_t^*||_2 = O_p\left(\frac{1}{\sqrt{N_0}}\right).
\]
\item[(c)]
\[
||(\hat{\bm A}' \hat{\bm A} + \delta_{N_0} \bm I_{T_0} )^{-1}\hat{\bm A}' - (\bm A' \bm A + \delta_{N_0} \bm I_{T_0} )^{-1}\bm A'||_2 
= O_p\left(\frac{1}{\delta_{N_0}\sqrt{N_0}}\right).
\]
\item[(d)]
\[
||(\hat{\bm A}' \hat{\bm A} + \delta_{N_0} \bm I_{T_0} )^{-1}\hat{\bm A}'\hat{\bm A} - (\bm A' \bm A + \delta_{N_0} \bm I_{T_0} )^{-1}\bm A'\bm A||_2 
= O_p\left(\frac{1}{\sqrt{\delta_{N_0} N_0}}\right).
\]
\end{enumerate}
\end{lemma}


\begin{proof}[Proof of Lemma \ref{lem:p1} (a)]
We prove the result for $\hat{\bm A}$ and $\bm A$.
By an analogous argument, we can prove the result for $\hat{\bm a}_t$ and $\bm a_t$.
Observe that
\[
||\hat{\bm A} - \bm A||_2 
\leq ||\hat{\bm W}^{1/2}||_2 ||\hat{\bm \Omega} - \bm \Omega||_2
+ ||\bm \Omega||_2 ||\hat{\bm W}^{1/2} - \bm W^{1/2}||_2.
\]
Under Assumption \ref{asp:4}, we have 
\[
||\hat{\bm W}^{1/2} - \bm W^{1/2}||_2 = O_p\left(\frac{1}{\sqrt{N_0}}\right) \quad\text{and}\quad
||\hat{\bm W}^{1/2}||_2 = O_p(1),
\]
and it suffices to show that
\begin{equation}\label{eq:b.a.3}
||\hat{\bm \Omega} - \bm \Omega||_2 = O_p\left(\frac{1}{\sqrt{N_0}}\right).
\end{equation}
Observe that, from the definitions,
\begin{eqnarray*}
|\hat{\Omega}_{j,-t} - \Omega_{j,-t}|
&\leq& \left|\frac{1}{N_0} \sum_{i=1}^{N_0} \omega_j(Z_i) Y_{i,-t}
- E[\omega_j(Z_i)  Y_{i,-t}|D_i=0] \right| \\
&& + \left| \frac{1}{N_0} \sum_{i=1}^{N_0} \omega_j(Z_i)Z_i' \hat{\beta}_{-t} - E[\omega_j(Z_i) Z_i'|D_i=0]\beta_{-t}\right|.
\end{eqnarray*}
Under Assumption \ref{asp:4}, applying the central limit theorem (CLT) gives
\begin{equation*}
\left|\frac{1}{N_0} \sum_{i=1}^{N_0} \omega_j(Z_i) Y_{i,-t}
- E[\omega_j(Z_i)  Y_{i,-t}|D_i=0]\right|
= O_p\left(\frac{1}{\sqrt{N_0}}\right).
\end{equation*}
For the second term, we have
\begin{eqnarray*}
&& \left| \frac{1}{N_0} \sum_{i=1}^{N_0} \omega_j(Z_i)Z_i' \hat{\beta}_{-t} - E[\omega_j(Z_i) Z_i'|D_i=0]\beta_{-t}\right| \\
&\leq& \left|\left|\frac{1}{N_0} \sum_{i=1}^{N_0} \omega_j(Z_i)Z_i'
- E[\omega_j(Z_i) Z_i'|D_i=0] \right|\right|_2 ||\hat \beta_{-t}||_2 \\ 
&& + ||E[\omega_j(Z_i) Z_i'|D_i=0]||_2 ||\hat{\beta}_{-t} - \beta_{-t}||_2.
\end{eqnarray*}
Under Assumption \ref{asp:4}, applying CLT and LLN gives
\begin{equation*}
\left|\left|\frac{1}{N_0} \sum_{i=1}^{N_0} \omega_j(Z_i)Z_i'
- E[\omega_j(Z_i) Z_i'|D_i=0] \right|\right|_2
= O_p\left(\frac{1}{\sqrt{N_0}}\right),
\end{equation*}
\begin{equation*}
||\hat{\beta}_{-t} - \beta_{-t}||_2 = O_p\left(\frac{1}{\sqrt{N_0}}\right), \quad\text{and hence}\quad
||\hat \beta_{-t}||_2 = O_p(1).
\end{equation*}
Therefore, we obtain 
\[
|\hat{\Omega}_{j,-t} - \Omega_{j,-t}| = O_p\left(\frac{1}{\sqrt{N_0}}\right),
\]
and this implies \eqref{eq:b.a.3}.
\end{proof}

\begin{proof}[Proof of Lemma \ref{lem:p1} (b)]
We define $\omega(Z_i)$ as $\omega(Z_i)\equiv(\omega_1(Z_i),\ldots,\omega_R(Z_i))'$.
For the first statement, observe that
\begin{eqnarray}
\bm \Omega_t &=& E[\omega(Z_i)\{(F^{*}_t)' \xi_{i,pre} + \epsilon_{it} - (F^{*}_t)' \bm \epsilon_{i,pre}\}|D_i=0] \nonumber \\
&=& \bm \Omega F^{*}_t + E[\omega(Z_i)\{(\epsilon_{it} - (F^{*}_t)' \bm \epsilon_{i,pre}\}|D_i=0]
= \bm \Omega F^{*}_t, \label{eq:b.b.1}
\end{eqnarray}
and hence the stated result follows.
We define
\[
\hat \xi_{it} \equiv Y_{i,t} - \hat{\beta}_{t}'Z_i
\quad\text{and}\quad
\hat \epsilon_{it} \equiv \epsilon_{it} - \hat{\nu}_t'Z_i,
\]
where $\hat{\nu}_t$ is defined as
\[
\hat{\nu}_t \ \equiv \ \left( \frac{1}{N_0} \sum_{i=1}^{N_0} Z_i Z_i' \right)^{-1} \frac{1}{N_0} \sum_{i=1}^{N_0} Z_i \epsilon_{it}.
\]
Then, as we derive \eqref{eq:a.1.3} in the proof of Theorem \ref{thm:1},
we obtain 
\[
\hat \xi_{it} = (F^{*}_t)' \hat \xi_{i,pre} + \hat \epsilon_{it} - (F^{*}_t)' \hat{\bm \epsilon}_{i,pre}.
\]
Hence, as in \eqref{eq:b.b.1}, observe that
\begin{eqnarray*}
\hat{\bm \Omega}_t &=& \hat{\bm \Omega} F^{*}_t +  
\frac{1}{N_0} \sum_{i=1}^{N_0} \omega(Z_i)\{\hat \epsilon_{it} - (F^{*}_t)' \hat{\bm \epsilon}_{i,pre}\},
\end{eqnarray*}
and it suffices to show that
\begin{equation}\label{eq:b.b.6}
\left|\left|\frac{1}{N_0} \sum_{i=1}^{N_0} \omega(Z_i)\{\hat \epsilon_{it} - (F^{*}_t)' \hat{\bm \epsilon}_{i,pre}\} \right|\right|_2
= O_p\left(\frac{1}{\sqrt{N_0}}\right).
\end{equation}
Observe that, from the definitions,
\begin{eqnarray*}
\left|\left|\frac{1}{N_0} \sum_{i=1}^{N_0} \omega(Z_i)\hat \epsilon_{it}\right|\right|_2
&\leq& \left|\left|\frac{1}{N_0} \sum_{i=1}^{N_0} \omega(Z_i) \epsilon_{it}\right|\right|_2 
+ \left|\left| \frac{1}{N_0} \sum_{i=1}^{N_0} \omega(Z_i)Z_i' \hat{\nu}_t \right|\right|_2
\end{eqnarray*}
Under Assumption \ref{asp:4}, applying the central limit theorem (CLT) gives
\begin{equation*}
\left|\left|\frac{1}{N_0} \sum_{i=1}^{N_0} \omega(Z_i) \epsilon_{it}\right|\right|_2
= O_p\left(\frac{1}{\sqrt{N_0}}\right).
\end{equation*}
For the second term, we have
\begin{eqnarray*}
\left|\left| \frac{1}{N_0} \sum_{i=1}^{N_0} \omega(Z_i)Z_i' \hat{\nu}_t \right|\right|_2 
\leq \left|\left|\frac{1}{N_0} \sum_{i=1}^{N_0} \omega_j(Z_i)Z_i'\right|\right|_2 ||\hat \nu_{t}||_2.
\end{eqnarray*}
Under Assumption \ref{asp:4}, applying CLT and LLN gives
\begin{equation*}
\left|\left|\frac{1}{N_0} \sum_{i=1}^{N_0} \omega_j(Z_i)Z_i'\right|\right|_2
= O_p\left(1\right)
\quad\text{and}\quad
||\hat \nu_{t}||_2  = O_p\left(\frac{1}{\sqrt{N_0}}\right).
\end{equation*}
Hence, we obtain 
\[
\left|\left|\frac{1}{N_0} \sum_{i=1}^{N_0} \omega(Z_i)\hat \epsilon_{it}\right|\right|_2 = O_p\left(\frac{1}{\sqrt{N_0}}\right),
\]
and this implies \eqref{eq:b.b.6}.
\end{proof}

\begin{proof}[Proof of Lemma \ref{lem:p1} (c)]
Observe that
\begin{eqnarray*}
&& ||(\hat{\bm A}' \hat{\bm A} + \delta_{N_0} \bm I_{T_0} )^{-1}\hat{\bm A}' - (\bm A' \bm A + \delta_{N_0} \bm I_{T_0} )^{-1}\bm A'||_2 \\
&\leq & ||(\hat{\bm A}' \hat{\bm A} + \delta_{N_0} \bm I_{T_0} )^{-1}
(\hat{\bm A} - \bm A)'||_2 \\
&& + ||\{(\hat{\bm A}' \hat{\bm A} + \delta_{N_0} \bm I_{T_0} )^{-1}
-(\bm A' \bm A + \delta_{N_0} \bm I_{T_0} )^{-1}\} \bm A' ||_2.
\end{eqnarray*}
From Lemma \ref{lem:a1} (a), we have
\[
||(\hat{\bm A}' \hat{\bm A} + \delta_{N_0} \bm I_{T_0} )^{-1}||_2 
\leq \frac{1}{\delta_{N_0}} =
O_p\left(\frac{1}{\delta_{N_0}}\right),
\]
and, for the first term, we have 
\begin{equation*}
||(\hat{\bm A}' \hat{\bm A} + \delta_{N_0} \bm I_{T_0} )^{-1}
(\hat{\bm A} - \bm A)'||_2
\leq ||(\hat{\bm A}' \hat{\bm A} + \delta_{N_0} \bm I_{T_0} )^{-1}||_2
||(\hat{\bm A} - \bm A)'||_2
= O_p\left(\frac{1}{\delta_{N_0}\sqrt{N_0}}\right),
\end{equation*}
and it suffices to show that
\begin{equation}\label{eq:b.p.17}
||\{(\hat{\bm A}' \hat{\bm A} + \delta_{N_0} \bm I_{T_0} )^{-1}
-(\bm A' \bm A + \delta_{N_0} \bm I_{T_0} )^{-1}\} \bm A' ||_2
= O_p\left(\frac{1}{\delta_{N_0}\sqrt{N_0}}\right).
\end{equation}
For this term, observe that
\begin{eqnarray*}
&& ||\{(\hat{\bm A}' \hat{\bm A} + \delta_{N_0} \bm I_{T_0} )^{-1}
-(\bm A' \bm A + \delta_{N_0} \bm I_{T_0} )^{-1}\} \bm A' ||_2 \\
&=& ||(\hat{\bm A}' \hat{\bm A} + \delta_{N_0} \bm I_{T_0} )^{-1}
\{(\bm A' \bm A + \delta_{N_0} \bm I_{T_0} ) 
- (\hat{\bm A}' \hat{\bm A} + \delta_{N_0} \bm I_{T_0} ) \}
(\bm A' \bm A + \delta_{N_0} \bm I_{T_0} )^{-1} \bm A' ||_2 \nonumber \\
&\leq & ||(\hat{\bm A}' \hat{\bm A} + \delta_{N_0} \bm I_{T_0} )^{-1}
\hat{\bm A}'(\bm A - \hat{\bm A})
(\bm A' \bm A + \delta_{N_0} \bm I_{T_0} )^{-1} \bm A' ||_2 \\
&& + ||(\hat{\bm A}' \hat{\bm A} + \delta_{N_0} \bm I_{T_0} )^{-1}
(\bm A - \hat{\bm A})'\bm A
(\bm A' \bm A + \delta_{N_0} \bm I_{T_0} )^{-1} \bm A' ||_2.
\end{eqnarray*}
From Lemma \ref{lem:a1} (b), we have
\[
||(\bm A' \bm A + \delta_{N_0} \bm I_{T_0} )^{-1}\bm A' ||_2 
= \frac{\sigma_r}{\sigma_r^2 + \delta_{N_0}} = O_p(1)
\]
and
\[
||(\hat{\bm A}' \hat{\bm A} + \delta_{N_0} \bm I_{T_0} )^{-1}\hat{\bm A}'||_2 
\leq \frac{1}{2\sqrt{\delta_{N_0}}}
=  O_p\left(\frac{1}{\sqrt{\delta_{N_0}}}\right),
\]
and hence we obtain
\begin{eqnarray*}
&& ||(\hat{\bm A}' \hat{\bm A} + \delta_{N_0} \bm I_{T_0} )^{-1}
\hat{\bm A}'(\bm A - \hat{\bm A})
(\bm A' \bm A + \delta_{N_0} \bm I_{T_0} )^{-1} \bm A' ||_2 \\
&\leq & ||(\hat{\bm A}' \hat{\bm A} + \delta_{N_0} \bm I_{T_0} )^{-1}
\hat{\bm A}'||_2
||(\bm A - \hat{\bm A})||_2
||(\bm A' \bm A + \delta_{N_0} \bm I_{T_0} )^{-1} \bm A' ||_2 
= O_p\left(\frac{1}{\sqrt{\delta_{N_0} N_0}}\right).
\end{eqnarray*}
From Lemma \ref{lem:a1} (c), we have
\[
||\bm A(\bm A' \bm A + \delta_{N_0} \bm I_{T_0} )^{-1}\bm A' ||_2 
=\frac{\sigma_1^2}{\sigma_1^2 + \delta_{N_0}}
=  O_p(1),
\]
and hence we obtain
\begin{eqnarray*}
&& ||(\hat{\bm A}' \hat{\bm A} + \delta_{N_0} \bm I_{T_0} )^{-1}
(\bm A - \hat{\bm A})'\bm A
(\bm A' \bm A + \delta_{N_0} \bm I_{T_0} )^{-1} \bm A' ||_2 \\
&\leq & ||(\hat{\bm A}' \hat{\bm A} + \delta_{N_0} \bm I_{T_0} )^{-1}||_2
||(\bm A - \hat{\bm A})'||_2
||\bm A(\bm A' \bm A + \delta_{N_0} \bm I_{T_0} )^{-1} \bm A' ||_2 
= O_p\left(\frac{1}{\delta_{N_0}\sqrt{N_0}}\right).
\end{eqnarray*}
Therefore, we obtain \eqref{eq:b.p.17}, and the stated result follows.
\end{proof}

\begin{proof}[Proof of Lemma \ref{lem:p1} (d)]
Observe that
\begin{eqnarray*}
&& (\hat{\bm A}' \hat{\bm A} + \delta_{N_0} \bm I_{T_0} )^{-1}
(\hat{\bm A}'\hat{\bm A} + \delta_{N_0} \bm I_{T_0} - \delta_{N_0} \bm I_{T_0}) \\
&& - (\bm A' \bm A + \delta_{N_0} \bm I_{T_0} )^{-1}
(\bm A'\bm A + \delta_{N_0} \bm I_{T_0} - \delta_{N_0} \bm I_{T_0}) \\
&=& \bm I_{T_0} - \delta_{N_0} (\hat{\bm A}' \hat{\bm A} + \delta_{N_0} \bm I_{T_0} )^{-1}
- \{\bm I_{T_0} - \delta_{N_0} (\bm A' \bm A + \delta_{N_0} \bm I_{T_0} )^{-1}\} \\
&=& \delta_{N_0} \{(\bm A' \bm A + \delta_{N_0} \bm I_{T_0} )^{-1}
- (\hat{\bm A}' \hat{\bm A} + \delta_{N_0} \bm I_{T_0} )^{-1}\},
\end{eqnarray*}
and it suffices to show that 
\begin{equation}\label{eq:b.d.4}
||(\hat{\bm A}' \hat{\bm A} + \delta_{N_0} \bm I_{T_0} )^{-1}
- (\bm A' \bm A + \delta_{N_0} \bm I_{T_0} )^{-1} ||_2 
= O_p\left(\frac{1}{\delta_{N_0}^{3/2}\sqrt{N_0}}\right).
\end{equation}
For this term, observe that
\begin{eqnarray*}
&& ||\{(\hat{\bm A}' \hat{\bm A} + \delta_{N_0} \bm I_{T_0} )^{-1}
-(\bm A' \bm A + \delta_{N_0} \bm I_{T_0} )^{-1}\} ||_2 \\
&=& ||(\hat{\bm A}' \hat{\bm A} + \delta_{N_0} \bm I_{T_0} )^{-1}
\{(\bm A' \bm A + \delta_{N_0} \bm I_{T_0} ) 
- (\hat{\bm A}' \hat{\bm A} + \delta_{N_0} \bm I_{T_0} ) \}
(\bm A' \bm A + \delta_{N_0} \bm I_{T_0} )^{-1} ||_2 \nonumber \\
&\leq & ||(\hat{\bm A}' \hat{\bm A} + \delta_{N_0} \bm I_{T_0} )^{-1}
\hat{\bm A}'(\bm A - \hat{\bm A})
(\bm A' \bm A + \delta_{N_0} \bm I_{T_0} )^{-1} ||_2 \\
&& + ||(\hat{\bm A}' \hat{\bm A} + \delta_{N_0} \bm I_{T_0} )^{-1}
(\bm A - \hat{\bm A})'\bm A
(\bm A' \bm A + \delta_{N_0} \bm I_{T_0} )^{-1} ||_2.
\end{eqnarray*}
From Lemma \ref{lem:a1} (b), we have
\[
||(\bm A' \bm A + \delta_{N_0} \bm I_{T_0} )^{-1}\bm A' ||_2 
= \frac{\sigma_r}{\sigma_r^2 + \delta_{N_0}} = O_p(1)
\]
and
\[
||(\hat{\bm A}' \hat{\bm A} + \delta_{N_0} \bm I_{T_0} )^{-1}\hat{\bm A}'||_2 
\leq \frac{1}{2\sqrt{\delta_{N_0}}}
=  O_p\left(\frac{1}{\sqrt{\delta_{N_0}}}\right),
\]
and hence we obtain
\begin{eqnarray*}
&& ||(\hat{\bm A}' \hat{\bm A} + \delta_{N_0} \bm I_{T_0} )^{-1}
\hat{\bm A}'(\bm A - \hat{\bm A})
(\bm A' \bm A + \delta_{N_0} \bm I_{T_0} )^{-1} ||_2 \\
&\leq & ||(\hat{\bm A}' \hat{\bm A} + \delta_{N_0} \bm I_{T_0} )^{-1}
\hat{\bm A}'||_2
||(\bm A - \hat{\bm A})||_2
||(\bm A' \bm A + \delta_{N_0} \bm I_{T_0} )^{-1} ||_2 
= O_p\left(\frac{1}{\delta_{N_0}^{3/2}\sqrt{ N_0}}\right).
\end{eqnarray*}
From Lemma \ref{lem:a1} (b), we have
\[
||\bm A(\bm A' \bm A + \delta_{N_0} \bm I_{T_0} )^{-1} ||_2 
= \frac{\sigma_r}{\sigma_r^2 + \delta_{N_0}} = O_p(1)
\]
and hence we obtain
\begin{eqnarray*}
&& ||(\hat{\bm A}' \hat{\bm A} + \delta_{N_0} \bm I_{T_0} )^{-1}
(\bm A - \hat{\bm A})'\bm A
(\bm A' \bm A + \delta_{N_0} \bm I_{T_0} )^{-1} ||_2 \\
&\leq & ||(\hat{\bm A}' \hat{\bm A} + \delta_{N_0} \bm I_{T_0} )^{-1}||_2
||(\bm A - \hat{\bm A})'||_2
||\bm A(\bm A' \bm A + \delta_{N_0} \bm I_{T_0} )^{-1} ||_2 
= O_p\left(\frac{1}{\delta_{N_0}\sqrt{N_0}}\right).
\end{eqnarray*}
Therefore, we obtain \eqref{eq:b.d.4}, and the stated result follows.
\end{proof}






\section{Auxiliary Lemma}\label{sec:c}
The following lemma provides some important properties of matrices used to derive the convergence rates of the terms in the proofs in Appendices \ref{sec:a} and \ref{sec:b}.
For a $n\times n$ square matrix $\bm D_1$, we write $\bm D_1=\diag(d_1,\ldots,d_n)$ when $\bm D_1$ is a diagonal matrix and $d_1,\ldots,d_n$ are the diagonal elements of $\bm D_1$.
For a $m\times n$ matrix $\bm D$, we write $\bm D=\diag(d_1,\ldots,d_l)$ with $l = \min\{m,n\}$ when $\bm D$ is expressed as either $\bm D=(\bm D_1,0)'$ or $\bm D=(\bm D_2,0)$, where $\bm D_1$ and $\bm D_2$ are $n\times n$ and $m\times m$ diagonal matrices, respectively.

\begin{lemma}\label{lem:a1}
Let $\bm B$ be a real valued $m\times n$ matrix, and let $\rank(\bm B)=q$ and $l = \min\{m,n\}$.
Let $\bm B = \bm U \bm \Sigma \bm V'$ be the singular value decomposition (SVD) of $\bm B$, where $\bm U\in \mathbb{R}^{m\times m}$ and $\bm V\in \mathbb{R}^{n\times n}$ are orthogonal matrices, $\bm \Sigma = \diag(\sigma_1(\bm B),\ldots,\sigma_l(\bm B))\in \mathbb{R}^{m\times n}$, and $\sigma_i(\bm B)\geq 0$ is the $i$th largest singular value of $\bm B$.
Then, for $\delta>0$, the following properties hold:
\begin{enumerate}
\item[(a)]
\[
||(\bm B' \bm B + \delta \bm I_{n} )^{-1}||_2 \leq \frac{1}{\delta}.
\]
\item[(b)] 
\[
||(\bm B' \bm B + \delta \bm I_{n} )^{-1}\bm B'||_2
= ||\bm B(\bm B' \bm B + \delta \bm I_{n} )^{-1}||_2
= \max_{j\in\{1,\ldots,l\}} \frac{\sigma_j(\bm B)}{\sigma_j(\bm B)^2 + \delta} \leq \frac{1}{2\sqrt{\delta}}.
\]
Specifically, if $\sqrt{\delta}<\sigma_q(\bm B)$,
\[
||(\bm B' \bm B + \delta \bm I_{n} )^{-1}\bm B'||_2
= ||\bm B(\bm B' \bm B + \delta \bm I_{n} )^{-1}||_2
= \frac{\sigma_q(\bm B)}{\sigma_q(\bm B)^2 + \delta}.
\]
\item[(c)]
\[
||\bm B (\bm B' \bm B + \delta \bm I_{n} )^{-1}\bm B'||_2 = \frac{\sigma_1(\bm B)^2}{\sigma_1(\bm B)^2 + \delta}.
\]
\item[(d)]
\[
||\bm B^+ - (\bm B' \bm B + \delta \bm I_{n} )^{-1}\bm B'||_2 = \frac{\delta}{\sigma_q(\bm B)(\sigma_q(\bm B)^2 + \delta)}.
\]
\end{enumerate}
\end{lemma}

We prove for the case of $m > n > q$. Other cases can be proved similarly.

\begin{proof}[Proof of Lemma \ref{lem:a1} (a)]
Using the properties of orthogonal matrices, $(\bm B' \bm B + \delta \bm I_{n} )^{-1}$ is expressed as
\begin{equation}\label{eq:c.1}
(\bm B' \bm B + \delta \bm I_{n} )^{-1}
= \{\bm V (\bm \Sigma' \bm \Sigma + \delta \bm I_{n}) \bm V'\}^{-1}
= \bm V (\bm \Sigma' \bm \Sigma + \delta \bm I_{n})^{-1} \bm V'.
\end{equation}
Observe that $\bm \Sigma' \bm \Sigma + \delta \bm I_{n}$ is expressed as
\begin{equation}
\bm \Sigma' \bm \Sigma + \delta \bm I_{n}
= \diag(\sigma_1(\bm B)^2 + \delta,\ldots,\sigma_n(\bm B)^2 + \delta),
\end{equation}
and hence $(\bm \Sigma' \bm \Sigma + \delta \bm I_{n})^{-1}$ is expressed as
\begin{equation}\label{eq:c.3}
(\bm \Sigma' \bm \Sigma + \delta \bm I_{n})^{-1}
= \diag\left(\frac{1}{\sigma_1(\bm B)^2 + \delta},\ldots,\frac{1}{\sigma_n(\bm B)^2 + \delta}\right).
\end{equation}
Therefore, we obtain
\begin{equation*}
||(\bm B' \bm B + \delta \bm I_{n} )^{-1}||_2 
= \frac{1}{\sigma_n(\bm B)^2 + \delta}
\leq \frac{1}{\delta}.
\end{equation*}
\end{proof}

\begin{proof}[Proof of Lemma \ref{lem:a1} (b)]
We prove the result for $(\bm B' \bm B + \delta \bm I_{n} )^{-1}\bm B'$.
We can prove the result for $\bm B(\bm B' \bm B + \delta \bm I_{n} )^{-1}$ by an analogous argument.
From \eqref{eq:c.1} and the properties of orthogonal matrices, $(\bm B' \bm B + \delta \bm I_{n} )^{-1}\bm B'$ is expressed as
\begin{equation}\label{eq:c.5}
(\bm B' \bm B + \delta \bm I_{n} )^{-1}\bm B'
= \bm V (\bm \Sigma' \bm \Sigma + \delta \bm I_{n})^{-1}\bm \Sigma' \bm U'.
\end{equation}
Observe that, using \eqref{eq:c.3}, $(\bm \Sigma' \bm \Sigma + \delta \bm I_{n})^{-1}\bm \Sigma'$ is expressed as
\begin{equation}\label{eq:c.6}
(\bm \Sigma' \bm \Sigma + \delta \bm I_{n})^{-1}\bm \Sigma'
= \diag\left(\frac{\sigma_1(\bm B)}{\sigma_1(\bm B)^2 + \delta},\ldots,\frac{\sigma_n(\bm B)}{\sigma_n(\bm B)^2 + \delta}\right).
\end{equation}
Therefore, we obtain
\[
||(\bm B' \bm B + \delta \bm I_{n} )^{-1}\bm B'||_2 = \max_{j\in\{1,\ldots,n\}} \frac{\sigma_j(\bm B)}{\sigma_j(\bm B)^2 + \delta}.
\]
For a function $g:\mathbb{R}\to \mathbb{R}$ defined as
\begin{equation*}
g(x) = \frac{x}{x^2 + \delta},
\end{equation*}
the derivative $g'(x)$ is calculated as
\begin{equation*}
g'(x) = \frac{-x^2 + \delta}{(x^2 + \delta)^2},
\end{equation*}
and hence $g'(x)\geq 0$ if $|x|\leq \sqrt{\delta}$ and $g'(x)< 0$ if $|x| > \sqrt{\delta}$.
Therefore, we have
\begin{equation*}
\max_{j\in\{1,\ldots,n\}} g(\sigma_j(\bm B)) \leq g(\sqrt{\delta})
= \frac{1}{2\sqrt{\delta}}.
\end{equation*}
When $\sqrt{\delta}<\sigma_q(\bm B)$, we have $g(\sigma_q(\bm B))\geq g(\sigma_i(\bm B))$ for each $i\in\{1,\ldots,n\}$, and we obtain
\begin{equation*}
||(\bm B' \bm B + \delta \bm I_{n} )^{-1}\bm B'||_2 
= g(\sigma_q(\bm B)).
\end{equation*}
\end{proof}

\begin{proof}[Proof of Lemma \ref{lem:a1} (c)]
From \eqref{eq:c.5} and the properties of orthogonal matrices, $\bm B(\bm B' \bm B + \delta \bm I_{n} )^{-1}\bm B'$ is expressed as
\begin{equation}\label{eq:c.10}
\bm B(\bm B' \bm B + \delta \bm I_{n} )^{-1}\bm B'
= \bm U \bm \Sigma (\bm \Sigma' \bm \Sigma + \delta \bm I_{n})^{-1}\bm \Sigma' \bm U'.
\end{equation}
Observe that, using \eqref{eq:c.6}, $\bm \Sigma (\bm \Sigma' \bm \Sigma + \delta \bm I_{n})^{-1}\bm \Sigma'$ is expressed as
\begin{equation}\label{eq:c.11}
\bm \Sigma (\bm \Sigma' \bm \Sigma + \delta \bm I_{n})^{-1}\bm \Sigma'
= \diag\left(\frac{\sigma_1(\bm B)^2}{\sigma_1(\bm B)^2 + \delta},\ldots,\frac{\sigma_n(\bm B)^2}{\sigma_n(\bm B)^2 + \delta},0,\ldots,0\right).
\end{equation}
Because a function $g_2:\mathbb{R}\to \mathbb{R}$ defined as
\begin{equation*}
g_2(x) = \frac{x^2}{x^2 + \delta} = 1 - \frac{\delta}{x^2 + \delta}
\end{equation*}
is an increasing function for $x\geq 0$, we obtain
\begin{equation*}
||\bm B(\bm B' \bm B + \delta \bm I_{n} )^{-1}\bm B'||_2 
= g_2(\sigma_1(\bm B)).
\end{equation*}
\end{proof}

\begin{proof}[Proof of Lemma \ref{lem:a1} (d)]
It is well known that the SVD of $\bm B^+$ is $\bm B^+ = \bm V \bm \Sigma^+ \bm U$, and $\bm \Sigma^+$ is expressed as 
\[
\bm \Sigma^+
= \diag\left(\frac{1}{\sigma_1(\bm B)},\ldots,\frac{1}{\sigma_q(\bm B)},0,\ldots,0\right).
\]
Then, using \eqref{eq:c.5}, $\bm B^+ - (\bm B' \bm B + \delta \bm I_{n} )^{-1}\bm B'$ is expressed as
\begin{equation*}
\bm B^+ - (\bm B' \bm B + \delta \bm I_{n} )^{-1}\bm B'
= \bm V \{\bm \Sigma^+ - (\bm \Sigma' \bm \Sigma + \delta \bm I_{n})^{-1}\bm \Sigma' \}\bm U'.
\end{equation*}
Observe that, using \eqref{eq:c.6}, $\bm \Sigma^+ - (\bm \Sigma' \bm \Sigma + \delta \bm I_{n})^{-1}\bm \Sigma'$ is expressed as
\begin{equation*}
\bm \Sigma^+ -  (\bm \Sigma' \bm \Sigma + \delta \bm I_{n})^{-1}\bm \Sigma'
= \diag\left(\frac{\delta}{\sigma_1(\bm B)(\sigma_1(\bm B)^2 + \delta)},\ldots,\frac{\delta}{\sigma_q(\bm B)(\sigma_q(\bm B)^2 + \delta)},0,\ldots,0\right).
\end{equation*}
Because a function $g_3:\mathbb{R}\to \mathbb{R}$ defined as
\begin{equation*}
g_3(x) = \frac{\delta}{x(x^2 + \delta)} 
\end{equation*}
is a decreasing function for $x > 0$, we obtain
\begin{equation*}
||\bm B^+ - \bm B(\bm B' \bm B + \delta \bm I_{n} )^{-1}\bm B'||_2 
= g_3(\sigma_q(\bm B)).
\end{equation*}
\end{proof}

\end{document}